\documentclass[a4paper]{article}
\usepackage[all]{xy}\usepackage[latin1]{inputenc}        
\usepackage[dvips]{graphics,graphicx}
\usepackage{amsfonts,amssymb,amsmath,color,mathrsfs, amstext}
\usepackage{amsbsy, amsopn, amscd, amsxtra, amsthm,authblk}
\usepackage{enumerate,algorithmicx,algorithm}
\usepackage{algpseudocode}
\usepackage{upref}
\usepackage{geometry}
\usepackage[displaymath]{lineno}
\usepackage{float}
\usepackage{yhmath}
\usepackage{booktabs}
\usepackage{siunitx}
\usepackage{bm}
\usepackage{multicol}
\usepackage{multirow}
\usepackage{makecell}
\usepackage{appendix}
\usepackage{exscale}
\usepackage{relsize}
\usepackage[colorlinks,
linkcolor=red,
anchorcolor=red,
citecolor=red
]{hyperref}

\numberwithin{equation}{section}

\DeclareMathOperator{\erf}{erf}
\DeclareMathOperator{\erfc}{erfc}

\newtheorem{theorem}{Theorem}
\newtheorem{lemma}[theorem]{Lemma}

\newtheorem{remark}[theorem]{Remark}
\newtheorem{proposition}[theorem]{Proposition}

\begin{document}
	
	\title{Random batch sum-of-Gaussians method for molecular dynamics simulation of particle systems in the NPT ensemble}
    
    \author[1]{Zhen Jiang\thanks{sanjohnson@sjtu.edu.cn}}
    \author[1,2]{Jiuyang Liang\thanks{jliang@flatironinstitute.org}}
\author[1]{Qi Zhou\thanks{zhouqi1729@sjtu.edu.cn}}

\affil[1]{School  of  Mathematical  Sciences, Shanghai  Jiao  Tong  University,  Shanghai 200240, China}
\affil[2]{Center for Computational Mathematics, Flatiron Institute, Simons Foundation, New York 10010, USA}
	\date{\today}
	
\maketitle
	
\begin{abstract}
In this work, we develop a random batch sum-of-Gaussians (RBSOG) method for molecular dynamics simulations of charged systems in the isothermal-isobaric (NPT) ensemble. We introduce an SOG splitting of the pressure-related $1/r^3$ kernel, yielding a smooth short-/long-range decomposition for instantaneous pressure evaluation. The long-range part is treated in Fourier space by random-batch importance sampling. Because the radial and non-radial pressure components favor different proposals, direct sampling either increases structure-factor evaluations and communication or leads to substantial variance inflation. To address this tradeoff, we introduce a measure-recalibration strategy that reuses Fourier modes drawn from the radial proposal and corrects them for the non-radial target, producing an unbiased pressure estimator with significantly reduced variance and negligible extra cost. The resulting method mitigates pressure artifacts caused by cutoff discontinuities in traditional Ewald-based treatments while preserving near-optimal $O(N)$ complexity. We provide theoretical evidence on pressure decomposition error, consistency of stochastic approximation, and convergence of RBSOG-based MD. Numerical experiments on bulk water, LiTFSI ionic liquids, and DPPC membranes show that RBSOG accurately reproduces key structural and dynamical observables with small batch sizes ($P\sim 100$). In large-scale benchmarks up to $10^7$ atoms on $2048$ CPU cores, RBSOG achieves about an order-of-magnitude speedup over particle-particle particle-mesh in electrostatic calculations for NPT simulations, together with a consistent $4\times$ variance reduction relative to random batch Ewald and excellent weak/strong scalability. Overall, RBSOG provides a practical and scalable route to reduce time-to-solution and communication cost in large-scale NPT simulations.  

\end{abstract}
	
\vspace{0.5cm}


\section{Introduction}
\label{sec::intro}

Molecular dynamics (MD) is a widely used simulation tool in modern science, with applications in chemistry, materials science, and biophysics \cite{Allen2017ComputerLiquids,karplus1990molecular,Axel1987Science}. In MD, particle trajectories are obtained by integrating equations of motion under a chosen statistical ensemble~\cite{Frenkel2001Understanding}. Ensemble averages over these trajectories yield estimates of structural, dynamical, and thermodynamic properties. In this work, we focus on the isothermal-isobaric (NPT) ensemble, which maintains a constant particle number $N$, pressure $P$, and temperature $T$. The NPT ensemble is a standard choice for mimicking experimental conditions in biomolecular simulations and is commonly used for solvated proteins, lipid membranes, viruses, and other biological systems~\cite{tobias1997atomic}. To support such applications, there remains a strong need for accurate, efficient, and scalable algorithms for NPT simulations.

For many systems of interest, electrostatic interactions are a major computational challenge because of their long-range nature~\cite{David2011Nanoscale,French2010Rev}. To address this issue, many fast algorithms have been developed. A large class of such methods is based on Ewald splitting~\cite{Ewald1921AnnPhys}, which decomposes the Coulomb kernel into a rapidly decaying short-range part and a smooth long-range part. The short-range part can be truncated in real space and summed directly, whereas the long-range part is evaluated in reciprocal space, leading to an algorithm with $O(N^{3/2})$ complexity. Algorithms used in mainstream software, such as particle-particle particle-mesh (PPPM) method~\cite{Hockney1988Computer}, particle-mesh Ewald (PME) method~\cite{Darden1993JCP}, and their variants~\cite{essmann1995smooth,shan2005gaussian,Marcello2016Press,DEShaw2020JCP}, use the fast Fourier transform (FFT) to accelerate the long-range calculation and reduce the complexity to $O(N\log N)$. Other fast methods include treecodes~\cite{Barnes1986Nature,lindsay2001particle}, multigrid methods~\cite{trottenberg2000multigrid}, and fast multipole methods (FMM)~\cite{greengard1987fast,greengard1987thesis}. However, many of these approaches rely heavily on global all-to-all communications and become limited at large-scale simulations~\cite{Arnold2013PRE}. For FFT-based schemes, six sequential rounds are needed to perform forward and backward Fourier transforms~\cite{ayala2021scalability}. Periodic FMMs~\cite{yan2018flexibly,pei2023fast} require partitioning the periodic tiling into far-field and near-field regions containing immediate neighboring cells, which requires additional computational and communication costs compared to their non-periodic version. Taken together, these factors make efficient and accurate NPT simulations challenging.

To address these scalability issues, the random batch Ewald (RBE) method~\cite{Jin2020SISC} was proposed. The original RBE formulation is tailored to the NVT ensemble: it avoids communication-intensive FFTs by using importance sampling over Fourier modes, and then integrates trajectories using unbiased long-range force estimators, yielding an $O(N)$ algorithm with strong parallel scalability~\cite{liang2022superscalability}. RBE was later extended to the NPT ensemble~\cite{liang2022iso}, but its performance in NPT is typically weaker than in NVT for two main reasons. First, the NPT implementation of RBE largely inherits the NVT sampling strategy for force evaluation; however, when computing the instantaneous pressure, this choice can lead to substantially higher variance at a fixed batch size. In practice, stable NPT simulations often require batch sizes on the order of $500\text{-}1000$, about $2.5\text{-}5\times$ larger than those used in NVT. Second, the Ewald splitting is discontinuous at the short-range cutoff, which affects both the force and the instantaneous pressure, resulting in significant truncation error and incorrect volume fluctuations. When particles cross the cutoff, this discontinuity may also produce artificial pressure jumps and introduce artifacts in the dynamics. For systems sensitive to pressure fluctuations, such as phospholipid bilayers under semi-isotropic pressure coupling, these artificial effects can amplify fluctuations in membrane area and bilayer spacing beyond their expected ranges~\cite{ahmed2010effect,kim2023neighbor}. These issues motivate new NPT algorithms that preserve high parallel scalability while reducing variance and improving stability.

In this paper, we develop a fast random batch sum-of-Gaussians (RBSOG) method for MD simulations in the NPT ensemble. RBSOG was originally introduced for long-range force calculations in NVT simulations~\cite{liang2023random}, where an SOG approximation of the $1/r$ kernel replaces Ewald splitting and removes force discontinuities induced by real-space cutoffs. Here, we extend this framework to pressure-tensor evaluation by constructing an SOG approximation of the $1/r^3$ kernel, yielding a high-order smooth pressure decomposition. A key NPT challenge is that the radial and non-radial pressure components favor different importance measures: sampling them independently increases structure-factor evaluations and communication, while forcing both components to share one proposal can significantly increase variance. To address this tradeoff, we introduce a measure-recalibration strategy that recycles Fourier modes sampled from the radial proposal and corrects them for the non-radial target, so that the same structure factors can be reused to construct an unbiased pressure estimator with substantially reduced variance and negligible extra cost. The resulting method preserves $O(N)$ complexity while improving variance reduction and parallel scalability. Numerical experiments show that, compared with PPPM, RBSOG achieves better weak and strong scaling and achieves more than an order-of-magnitude speedup in pressure evaluation in large-scale benchmarks with up to $10^7$ atoms and 2048 cores. Relative to RBE, it reduces pressure variance by about a factor of $4$ at the same batch size and suppresses nonphysical surface fluctuations in semi-isotropic membrane simulations. The framework is broadly applicable to NPT simulations and can be extended to other interaction kernels by combining it with kernel-independent SOG techniques~\cite{greengard2018anisotropic,gao2022kernel,lin2025weighted}.

The rest of the paper is organized as follows. In Section~\ref{sec::Prelim}, we review the necessary preliminaries for the NPT ensemble and kernel decompositions. Section~\ref{sec::RBSOG_NPT} presents the RBSOG method in detail, including SOG-based pressure-tensor formulas, the random batch importance sampling with measure recalibration strategy, and the associated decomposition error and MD convergence analysis. In Section~\ref{sec::NumResult}, we report systematic numerical results to benchmark the accuracy of RBSOG-based NPT simulations, and compare it with the PPPM and RBE methods in terms of CPU time per step, strong and weak parallel scalability, and variance reduction. Concluding remarks are given in Section~\ref{sec::Conclusion}.

\section{Preliminaries}
\label{sec::Prelim}

In this section, we review preliminaries on the NPT ensemble, the evaluation of the instantaneous pressure tensor, and kernel-decomposition techniques for long-range electrostatic interactions. These elements form the theoretical
foundation of this work.

\subsection{NPT ensemble and instantaneous pressure}
\label{sec::Distribution}
The NPT ensemble describes a system with fixed particle number $N$ and temperature $T$, subject to an external pressure $P$. We parameterize the simulation cell $\Omega$ by the tensor
\begin{equation}
\bm{h}=[\bm{h}_1,\bm{h}_2,\bm{h}_3]\in \mathbb{R}^{3\times 3},
\end{equation}
whose determinant $V:=\det(\bm{h})$ gives the cell volume. It is convenient to separate volume and shape by writing
$\bm{h}=V^{1/3}\bm{h}_0$, where $\det(\bm{h}_0)=1$ so that $\bm{h}_0$ encodes the cell shape. A microscopic state $S$ in the NPT ensemble has statistical weight governed by the Boltzmann distribution~\cite{Frenkel2001Understanding}
\begin{equation}
\label{eq::NPT}
\mathcal{P}(S)\propto \exp\!\left[-\beta\,(E+PV)\right],
\end{equation}
where $E$ is the total energy, and $\beta=(k_BT)^{-1}$ with $k_B$ the Boltzmann constant.

For a system of $N$ particles, let $\bm{r}=\{\bm{r}_i\}_{i=1}^N$ and $\bm{p}=\{\bm{p}_i\}_{i=1}^N$ denote positions and momenta. For a fixed cell shape $\bm{h}_0$, the NPT partition function can be written as~\cite{Frenkel2001Understanding}
\begin{equation}
\label{eq::DistFuncNPT}
\Delta(N,P,T)
=\int e^{-\beta(E+PV)}\,\mathrm{d}\bm{r}\,\mathrm{d}\bm{p}\,\mathrm{d}V
=\int_{0}^{\infty} e^{-\beta PV}\, Q(N,V,T)\,\mathrm{d}V,
\end{equation}
where $\mathrm{d}\bm{r}=\mathrm{d}\bm{r}_1\cdots \mathrm{d}\bm{r}_N$ and $\mathrm{d}\bm{p}=\mathrm{d}\bm{p}_1\cdots \mathrm{d}\bm{p}_N$, and
\begin{equation}
\label{eq::DistFuncNVT}
Q(N,V,T)=\int_{\Omega(V,\bm{h}_0)\times \mathbb{R}^{3N}} e^{-\beta E(\bm{r},\bm{p})}\,\mathrm{d}\bm{r}\,\mathrm{d}\bm{p}
\end{equation}
is the canonical (NVT) partition function at fixed volume. For a fully flexible cell, where the shape $\bm{h}_0$ is allowed to fluctuate, the NPT partition function can be formulated as
\begin{equation}
\label{eq::h_DistFuncNPT}
\Delta(N,P,T)=\int_{\det(\bm{h})>0}\det(\bm{h})^{-2}\,e^{-\beta P\det(\bm{h})}\,Q(N,\bm{h},T)\,\mathrm{d}\bm{h},
\end{equation}
where $\{\bm{h}:\det(\bm{h})>0\}\subset \mathbb{R}^{3\times 3}$ and $\mathrm{d}\bm{h}$ denotes integration over the independent entries of $\bm{h}$. The factor $\det(\bm{h})^{-2}$ arises from the corresponding Jacobian in the change of variables used to express the measure. In this setting, it is convenient to introduce scaled coordinates $\bm{s}_i=\bm{h}^{-1}\bm{r}_i$ (so that $\bm{r}_i=\bm{h}\bm{s}_i$) and the conjugate momenta $\bm{p}_i^{s}=\bm{h}^{\top}\bm{p}_i$. With these variables, the NVT partition function $Q(N,\bm{h},T)$ can be expressed as
\begin{equation}
\label{eq::Q_NHT}
Q(N,\bm{h},T)=\int \exp\!\left[-\beta\,(K+U)\right]\,\mathrm{d}\bm{s}\,\mathrm{d}\bm{p}^{s},
\end{equation}
where $K$ and $U$ denote the kinetic and potential energies, respectively, and
$\mathrm{d}\bm{s}=\mathrm{d}\bm{s}_1\cdots \mathrm{d}\bm{s}_N$ and $\mathrm{d}\bm{p}^{s}=\mathrm{d}\bm{p}_1^{s}\cdots \mathrm{d}\bm{p}_N^{s}$.
We refer to~\cite{Frenkel2001Understanding,LeimkuhlerMatthews2015MD} for detailed derivations and further discussion.

According to classical statistical mechanics, the (macroscopic) pressure tensor
$\widetilde{\bm{P}}$ for a general simulation cell can be obtained from the NVT
partition function $Q(N,\bm{h},T)$ as
\begin{equation}
\label{eq::MacroPress_re}
\widetilde{\bm{P}}
=\frac{1}{\beta \det(\bm{h})}
\left(\frac{\partial \log Q(N,\bm{h},T)}{\partial \bm{h}}\right)\bm{h}^{\top}.
\end{equation}
Using Eq.~\eqref{eq::Q_NHT}, the components of $\widetilde{\bm{P}}$ can be written as
\begin{equation}
\label{eq::FullyPress_re}
\widetilde{P}_{\mu \nu}
=\frac{1}{Q(N,\bm{h},T)}
\int \exp\!\left[-\beta\left(K+U\right)\right]
\left[
\frac{1}{\det(\bm{h})}\sum_{\eta=1}^{3}
\frac{\partial(-K-U)}{\partial h_{\mu\eta}}\,h_{\nu\eta}
\right]
\,\mathrm{d}\boldsymbol{s}\,\mathrm{d}\boldsymbol{p}^{\boldsymbol{s}},
\end{equation}
where $\mu,\nu\in\{1,2,3\}$. Eq.~\eqref{eq::FullyPress_re} is an ensemble
average, and thus implies $\widetilde{\bm{P}}=\langle \bm{P}_{\mathrm{ins}}\rangle$,
where the instantaneous pressure tensor is defined by
\begin{equation}
\label{eq::PresTensor_re}
\bm{P}_{\mathrm{ins}}
=\frac{1}{\det(\bm{h})}\frac{\partial(-K-U)}{\partial \bm{h}}\bm{h}^{\top}
:=\bm{P}_{\mathrm{ins}}^{K}+\bm{P}_{\mathrm{ins}}^{U}.
\end{equation}
The kinetic contribution is
\begin{equation}
\label{eq::K_U_Pins_re}
\bm{P}_{\mathrm{ins}}^{K}
=\frac{1}{\det(\bm{h})}\sum_{i=1}^{N}\frac{\bm{p}_i\otimes \bm{p}_i}{m_i},
\end{equation}
and the potential contribution can be expressed as
\begin{equation}
\label{eq::potential_re}
\begin{split}
\bm{P}_{\mathrm{ins}}^{U}
&=-\frac{1}{\det(\bm{h})}\frac{\partial U(\{\bm{h}\bm{s}_i\};\bm{h})}{\partial \bm{h}}\bm{h}^{\top}\\
&=\frac{1}{\det(\bm{h})}
\left[
\sum_{i=1}^{N}\bm{F}_i\otimes\bm{r}_i
-\left(\frac{\partial U(\{\bm{r}_i\};\bm{h})}{\partial \bm{h}}\right)\bm{h}^{\top}
\right],
\end{split}
\end{equation}
where $\bm{F}_i=-\nabla_{\bm{r}_i}U$ is the force and $\otimes$ denotes the outer
product. The second term in the last line of Eq.~\eqref{eq::potential_re} accounts
for any explicit dependence of the potential energy on the cell tensor (e.g., due
to the minimum-image convention or other aspects of the periodic geometry).

\subsection{Kernel decomposition for electrostatic interactions}
\label{sec::useries}

In NPT simulations, the main computational cost comes from repeatedly evaluating the
potential energy and its contribution to the instantaneous pressure tensor.
For charged systems, this step is particularly expensive because Coulomb
interactions decay slowly and must be summed over periodic images. Let $q_i$ denote
the charge of particle $i$. Under periodic boundary conditions specified by the cell
matrix $\bm{h}$, the Coulomb energy can be written as
\begin{equation}\label{eq::potential}
U_{\text{coul}}=\frac{1}{2}\sum_{i,j=1}^{N}\sum_{\bm{n}\in\mathbb{Z}^3}{}^{\prime}\frac{q_{i}q_{j}}{\left|\bm{r}_{ij}+\bm{h}\bm{n}\right|},
\end{equation}
where $\bm{r}_{ij}:=\bm{r}_i-\bm{r}_j$ and the prime indicates that the case of $i=j$ and $\bm{n}=\bm{0}$ is excluded in the double sum. We assume charge neutrality, i.e., $\sum_i q_i=0$. The lattice sum in
Eq.~\eqref{eq::potential} is only conditionally convergent, meaning that its value
depends on the summation order. As a result, a direct real-space truncation is not
reliable and can also introduce force discontinuities. 

To address this issue, the classical Ewald decomposition~\cite{Ewald1921AnnPhys} is frequently used, which splits the Coulomb kernel into two components:
\begin{equation}\label{eq::Ewald}
\frac{1}{r}=\frac{1}{r}\erfc\left(\frac{r}{\sqrt{2}\alpha}\right)+\frac{1}{r}\erf\left(\frac{r}{\sqrt{2}\alpha}\right),
\end{equation}
where $\alpha>0$ is a splitting factor. The first term decays rapidly and can
be truncated at a cutoff radius $r_c$ in real space, while the second term is smooth
and is evaluated efficiently in Fourier space. In practice, FFT-based solvers (e.g.,
PME/PPPM) accelerate the Fourier-space contribution, reducing the per-step cost
to $O(N\log N)$ for systems of size $N$. The same splitting strategy extends to pressure calculations. Applying
Eq.~\eqref{eq::potential_re} to $U_{\text{coul}}$ yields the Coulomb contribution to
the instantaneous pressure tensor
\begin{equation}\label{eq::Pcoul}
\bm{P}_{\text{coul}}
=\frac{1}{2\det(\bm{h})}
\sum_{\bm{n}\in\mathbb{Z}^3}{}^{\prime}\sum_{i,j=1}^{N}
q_iq_j\,
\frac{\left(\bm{r}_{ij}+\bm{h}\bm{n}\right)\otimes\left(\bm{r}_{ij}+\bm{h}\bm{n}\right)}
{\left|\bm{r}_{ij}+\bm{h}\bm{n}\right|^{3}},
\end{equation}
whose radial scaling is governed by the $1/r^3$ kernel. Ewald-type methods provide consistent real- and Fourier-space treatments for both the energy kernel $1/r$ and the pressure-related kernel $1/r^3$, while preserving the same asymptotic complexity for pressure evaluation~\cite{di2015stochastic,Marcello2016Press}. These ideas form
the basis of widely used implementations in MD packages such as
LAMMPS~\cite{thompson2022lammps} and GROMACS~\cite{abraham2015gromacs}.

Unlike Ewald splitting, the SOG decomposition is a
more recent approach that approximates $1/r^\beta$ by a finite sum of smooth Gaussians and uses this representation to construct a real-/Fourier-space split. For $\beta>0$, the inverse-power kernel admits the integral identity
\begin{equation}
\label{eq::BSAInt}
\frac{1}{r^\beta}
=\frac{1}{\Gamma(\beta/2)}
\int_{-\infty}^{\infty}
e^{-e^t r^2+\frac{\beta}{2}t}\,\mathrm{d}t,
\end{equation}
where $\Gamma(\cdot)$ is the Gamma function. Discretizing the integral by the uniform trapezoidal rule with nodes $t_\ell=t_0+\ell h$, starting point $t_0=-\log(2\sigma^2)$, and step size $h=2\log b$ (with $\sigma>0$ and $b>1$) yields the bilateral series approximation (BSA)~\cite{beylkin2005approximation,beylkin2010approximation}:
\begin{equation}
\label{eq::BSA_beta}
\frac{1}{r^{\beta}}
\approx
\frac{2 \log b}{(\sqrt{2}\sigma)^{\beta}\Gamma(\beta/2)}
\sum_{\ell=-\infty}^{\infty} b^{-\beta\ell}
\exp\!\left[-\frac{1}{2}\left(\frac{r}{b^{\ell} \sigma}\right)^2\right].
\end{equation}
As $b\rightarrow 1^{+}$, the BSA achieves an exponentially small relative error
\cite{DEShaw2020JCP}:
\begin{equation}
\label{eq::BSA_err}
\left|
1-
\frac{2 r^{\beta}\log b}{(\sqrt{2}\sigma)^{\beta}\Gamma(\beta/2)}
\sum_{\ell=-\infty}^{\infty} b^{-\beta\ell}
\exp\!\left[-\frac{1}{2}\left(\frac{r}{b^{\ell} \sigma}\right)^2\right]
\right|
\lesssim
\frac{2\sqrt{2\pi}}{\Gamma(\beta/2)}
\left(\frac{\beta^2}{4}+\frac{\pi^2}{(\log b)^2}\right)^{\frac{\beta-1}{4}}
e^{-\frac{\pi^2}{2\log b}},
\end{equation}
where the dominant exponential factor is independent of $\beta$. The so-called $u$-series method~\cite{DEShaw2020JCP} is obtained by specializing
Eq.~\eqref{eq::BSA_beta} to $\beta=1$ and truncating the Gaussian series at $0\leq \ell \leq M-1$ with $M$ the number of retained Gaussians. It splits the Coulomb kernel into short- and long-range components,
\begin{equation}
\label{eq::potential_split}
\frac{1}{r}=\mathcal{N}_b^\sigma(r)+\mathcal{F}_b^\sigma(r),
\end{equation}
where the long-range component collects Gaussians with larger bandwidths (slower
decay) and is evaluated in Fourier space,
\begin{equation}
\label{eq::long}
\mathcal{F}_b^\sigma(r)
=\sum_{\ell=0}^{M-1} w_\ell \exp\!\left(-\frac{r^2}{s_\ell^2}\right),
\end{equation}
with $w_\ell=(\pi/2)^{-1/2}\sigma^{-1}b^{-\ell}\log b$ and $s_\ell=\sqrt{2}b^{\ell}\sigma$. The complementary short-range term is
\begin{equation}
\label{eq::short}
\mathcal{N}_b^\sigma(r)=
\begin{cases}
\dfrac{1}{r}-\mathcal{F}_b^\sigma(r), & r<r_c,\\[1em]
0, & r\ge r_c,
\end{cases}
\end{equation}
which decays rapidly and can be evaluated efficiently in real space. A key advantage of this SOG decomposition is that the cutoff behavior can be tuned to enforce smoothness at $r_c$. We choose $r_c$ as the smallest root of
\begin{equation}\label{Eq::2.20}
\frac{1}{r}-\mathcal{F}^{\sigma}_b(r)=0,
\end{equation}
which enforces $C^0$ continuity. For stable MD simulations and accurate pressure statistics, it is also desirable to enforce $C^1$ continuity at the cutoff, i.e.,
continuity of the force:
\begin{equation}\label{eq::2.22}
\frac{\mathrm{d}}{\mathrm{d}r}\left[\frac{1}{r}-\mathcal{F}_{b}^{\sigma}(r)\right]\bigg|_{r=r_c}=0.
\end{equation}
To satisfy Eqs.~\eqref{Eq::2.20}--\eqref{eq::2.22} simultaneously, we introduce an
additional scaling parameter $\omega$ for the narrowest Gaussian,
\begin{equation}
w_0 \leftarrow \omega \left(\frac{\pi}{2}\right)^{-1/2}\sigma^{-1}\log b,
\end{equation}
and solve for $(r_c,\omega)$ given $(b,\sigma,M)$. This SOG-based splitting has been used for accurate electrostatic energy and force evaluation~\cite{DEShaw2020JCP,liang2023random,gao2024fast,chen2025random,chen2026random}, and has also been adopted in machine-learning interatomic potentials~\cite{ji2025machine,ji2026accurate} and in numerical solvers for Schr\"odinger-type equations~\cite{zhou2025sum}. For completeness, the corresponding energy and force decompositions induced by Eq.~\eqref{eq::potential_split} are summarized in~\ref{app::force_energy}.

\section{Random batch sum-of-Gaussians method for NPT ensemble}
\label{sec::RBSOG_NPT}

In this section, we extend the SOG framework to the pressure-related $1/r^3$ kernel and derive the corresponding pressure-tensor formulation. We then present the RBSOG method in detail, including random batch importance sampling and a measure-recalibration procedure to reduce sampling-induced variance. Together, these components yield an efficient framework for NPT-ensemble simulations.

\subsection{SOG decomposition for instantaneous pressure} 
\label{sec::Pressure}
We begin by setting $\beta=3$ in Eq.~\eqref{eq::BSA_beta}, which yields an SOG approximation of the pressure-related $1/r^3$ kernel:
\begin{equation}\label{eq::approxerror}
\frac{1}{r^3}\approx \frac{\sqrt{2}\,\log \widetilde{b}}{\sqrt{\pi}\,\widetilde{\sigma}^3}
\sum_{\ell=-\infty}^{\infty} \widetilde{b}^{-3\ell}
\exp\!\left[-\frac{1}{2}\left(\frac{r}{\widetilde{b}^{\ell}\widetilde{\sigma}}\right)^2\right],
\end{equation}
where $\widetilde{b}>1$ and $\widetilde{\sigma}>0$ are splitting parameters, analogous to $b$ and $\sigma$ in the potential splitting.  Eq.~\eqref{eq::BSA_err} implies that this approximation has uniform error bound. For notational convenience, we define
\[
\widetilde{w}_\ell=\left(\frac{\pi}{2}\right)^{-1/2}\frac{\log \widetilde{b}}{\widetilde{b}^{3\ell}\widetilde{\sigma}^3},
\qquad
\widetilde{s}_\ell=\sqrt{2}\,\widetilde{b}^{\ell}\widetilde{\sigma}.
\]
By truncating the Gaussian sum at $0\leq \ell\leq \widetilde{M}-1$, we split the kernel into two components:
\begin{equation}\label{eq::SOGdecompositionPress}
\frac{1}{r^3}=\widetilde{\mathcal{N}}_{b}^\sigma(r)+\widetilde{\mathcal{F}}_b^\sigma(r),
\end{equation}
where 
\begin{equation}
    \label{eq::pressure_split} \widetilde{\mathcal{F}}_b^\sigma(r)=\sum_{\ell=0}^{\widetilde{M}-1}\widetilde{w}_\ell e^{-r^2/\widetilde{s}_{\ell}^2}\,\quad \text{and}\quad\,\widetilde{\mathcal{N}}_b^\sigma(r)= \begin{cases}1 / r^3 -\widetilde{\mathcal{F}}_b^\sigma(r), & \text { if } r<r_c \\[1em] 0, & \text { if } r \geq r_c\end{cases}
\end{equation}
are long- and short-range contributions, respectively. 

By substituting the SOG decomposition in Eq.~\eqref{eq::SOGdecompositionPress} into Eq.~\eqref{eq::Pcoul}, we split the Coulomb contribution to the pressure tensor into two parts:
\begin{equation}\label{eq::SOGDecomp}
\bm{P}_{\text{coul}}:=\bm{P}_{\text{coul}}^{\mathcal{N}}+\bm{P}_{\text{coul}}^{\mathcal{F}},
\end{equation}
where
\begin{equation}
\label{eq::P_short}
\begin{aligned}
	\bm{P}_{\text{coul}}^{\mathcal{N}} &=\frac{1}{2\det(\bm{h})}\sum_{\bm{n}}\!^\prime\sum_{i,j}q_iq_j\widetilde{\mathcal{N}}_b^{\sigma}(|\bm{r}_{ij}+\bm{h}\bm{n}|)(\bm{r}_{ij}+\bm{h}\bm{n})\otimes (\bm{r}_{ij}+\bm{h}\bm{n})
\end{aligned}
\end{equation}
and
\begin{equation}\label{eq::PFcoul}
\bm{P}_{\text{coul}}^{\mathcal{F}} =\frac{1}{2\det(\bm{h})}\sum_{\bm{n}}\sum_{i,j}q_iq_j\widetilde{\mathcal{F}}_b^{\sigma}(|\bm{r}_{ij}+\bm{h}\bm{n}|)(\bm{r}_{ij}+\bm{h}\bm{n})\otimes (\bm{r}_{ij}+\bm{h}\bm{n})
\end{equation}
are the short-range and long-range parts, respectively. Because $\widetilde{\mathcal{N}}_{b}^{\sigma}$ decays rapidly, $\bm{P}_{\text{coul}}^{\mathcal{N}}$ is computed by real-space truncation. Since $\widetilde{\mathcal{F}}_{b}^{\sigma}(r)$ is smooth, $\bm{P}_{\text{coul}}^{\mathcal{F}}$ can be evaluated efficiently in the Fourier space. We define the Fourier transform pair as
\begin{equation}\label{eq::FourierTransform}
\hat{f}(\bm{k})=\int_{\bm{h}}f(\bm{r})e^{-i\bm{k}\cdot\bm{r}}\mathrm{d}\bm{r},\quad f(\bm{r})=\frac{1}{\det(\bm{h})}\sum_{\bm{k}} \hat{f}(\bm{k})e^{i\bm{k}\cdot \bm{r}},
\end{equation}
with $\bm{k}= 2\pi \bm{h}^{-\top}\bm{m}$, $\bm{m}\in\mathbb{Z}^3$, and $\bm{h}^{-\top}=(\bm{h}^{-1})^{\top}$. The corresponding $\bm{k}$-space form of $\bm{P}_{\text{coul}}^{\mathcal{F}}$ is given below. 

\begin{theorem}
Assume tinfoil boundary conditions. The Fourier space contribution to the Coulombic pressure tensor, $\bm{P}_{\mathrm{coul}}^{\mathcal{F}}$, admits the representation
\begin{equation}
\label{eq::P_long}
\bm{P}_{\mathrm{coul}}^{\mathcal{F}}
=\frac{\pi^{3/2}}{4\det(\bm{h})^2}
\sum_{\bm{k}\neq \bm{0}}
\sum_{\ell=0}^{\widetilde{M}-1}
\widetilde{w}_{\ell}\widetilde{s}_{\ell}^{5}
e^{-\widetilde{s}_{\ell}^2k^2/4}
\left(\bm{I}-\frac{\widetilde{s}_{\ell}^2}{2}\bm{k}\otimes\bm{k}\right)
\left|\rho(\bm{k})\right|^2,
\end{equation}
where
\begin{equation}
\label{eq::rho_k}
\rho(\bm{k})
=\sum_{i=1}^{N} q_i e^{i\bm{k}\cdot \bm{r}_i}
=\sum_{i=1}^{N} q_i e^{2\pi i\bm{m}\cdot \bm{s}_i}
\end{equation}
is the charge structure factor. In particular, $\rho(\bm{k})$ depends only on particle charges and fractional coordinates, and is independent of the cell tensor $\bm{h}$.
\end{theorem}

\begin{proof}
Define the tensor-valued function
\begin{equation}
\bm{g}_{b}^{\sigma}(\bm{r})
:=\widetilde{\mathcal{F}}_{b}^{\sigma}(|\bm{r}|)
\,\bm{r}\otimes\bm{r}.
\end{equation}
Applying the Poisson summation formula to Eq.~\eqref{eq::PFcoul} yields
\begin{equation}\label{eq::PcoulF}
\bm{P}_{\text{coul}}^{\mathcal{F}}
=\frac{1}{2\det(\bm{h})^2}
\sum_{\bm{k}}
\widehat{\bm{g}}_{b}^{\sigma}(\bm{k})
\sum_{i,j} q_i q_j\, e^{i\bm{k}\cdot\bm{r}_{ij}},
\end{equation}
where $\widehat{\bm{g}}_{b}^{\sigma}(\bm{k})$ denotes the Fourier transform of $\bm{g}_{b}^{\sigma}(\bm{r})$.
Using the Fourier transform convention in Eq.~\eqref{eq::FourierTransform} and evaluating component-wise, we obtain
\begin{equation}\label{eq::gbsigma}
\widehat{\bm{g}}_{b}^{\sigma}(\bm{k})
=\sum_{\ell=0}^{\widetilde{M}-1}
\frac{1}{2}\pi^{3/2}\widetilde{w}_\ell \widetilde{s}_\ell^{5}
e^{-\widetilde{s}_\ell^2k^2/4}
\left(\bm{I}-\frac{\widetilde{s}_\ell^2}{2}\bm{k}\otimes\bm{k}\right).
\end{equation}
Substituting Eq.~\eqref{eq::gbsigma} into Eq.~\eqref{eq::PcoulF} and removing the $\bm{k}=\bm{0}$ contribution (due to tinfoil boundary conditions) completes the proof.
\end{proof}

Smooth pressure decomposition is important for stable NPT simulations. Although Eq.~\eqref{eq::SOGdecompositionPress} is analytic, truncation at the cutoff $r_c$ can still introduce discontinuity and unknown artifacts. Ewald decomposition suffers from a similar issue, and previous studies~\cite{kim2023neighbor,kubincova2020reaction} have shown that cutoff artifacts can significantly distort volume fluctuations. In the SOG decomposition, this issue can be mitigated by parameter calibration. Specifically, we rescale the narrowest Gaussian component in \(\widetilde{\mathcal{F}}_{b}^{\sigma}(r)\) as
\begin{equation}\label{omega0}
\widetilde{w}_{0}\;\rightarrow\; \widetilde{\omega}\widetilde{w}_{0}=\widetilde{\omega}\left(\frac{\pi}{2}\right)^{-1/2}\widetilde{\sigma}^{-3}\log \widetilde{b},
\end{equation}
and choose $\widetilde{\omega}$ to enforce $C^{0}$ continuity at $r_c$, i.e.,
\begin{equation}\label{eq::continuity}
\frac{1}{r_{c}^3}-\widetilde{\mathcal{F}}_{b}^{\sigma}(r_c)=0.
\end{equation}
If higher-order smoothness is needed, one may further tune the bandwidth of the narrowest Gaussian to satisfy $C^{1}$ continuity (and analogously for higher derivatives). In our simulations, $C^{0}$ continuity is sufficient to maintain stable cell-volume fluctuations.

Another key requirement is consistency with the virial theorem~\cite{Frenkel2001Understanding}.For Coulomb interactions, the electrostatic energy is homogeneous under uniform scaling: for any $a>0$,
\[
U(\{a\bm{r}_{i}\};a\bm{h})=\frac{1}{a}\,U(\{\bm{r}_{i}\};\bm{h}).
\]
Therefore,
$\mathrm{tr}\!\left(\bm{P}_{\mathrm{coul}}\right)=U_{\text{coul}}/\det(\bm{h})$.
Within the SOG framework, this identity is preserved when the potential and pressure decompositions are truncated consistently, so the resulting pressure remains unbiased. The formal statement is given in Theorem~\ref{thm::virial}; a detailed proof is provided in~\ref{app::virialtheorem}.

\begin{theorem}
\label{thm::virial}
Suppose the SOG decompositions of the potential and pressure in
Eqs.~\eqref{eq::couldecomp} and \eqref{eq::SOGDecomp} are truncated using the same number of terms, i.e., $M=\widetilde{M}$, and the same splitting parameters $b=\widetilde{b}$ and $\sigma=\widetilde{\sigma}$ are used in both decompositions. Then the virial identity
\begin{equation}
\label{eq::Virli_flexible}
\mathrm{tr}\!\left(\bm{P}^{\mathcal{N}}_{\mathrm{coul}}+\bm{P}^{\mathcal{F}}_{\mathrm{coul}}\right)
=\frac{U_{\mathrm{coul}}^\mathcal{N}+U_{\mathrm{coul}}^\mathcal{F}-U_{\mathrm{self}}}{\det(\bm{h})}
\end{equation}
is preserved by the truncated SOG formulation. Moreover, under this consistent truncation, the matching of the narrowest-Gaussian weights implies $\omega=\widetilde{\omega}$, so it enforces continuity of both the force and pressure with the same set of splitting parameters. 
\end{theorem}

Enforcing $M=\widetilde{M}$, $b=\widetilde{b}$, and $\sigma=\widetilde{\sigma}$ is not always a free lunch. Table~\ref{table::para_error} shows that energy, force, and pressure can require different numbers of Gaussian terms at the same tolerance. In practice, for force/pressure tolerances around $10^{-3}\sim10^{-4}$, the required term counts are close. At higher precision, one may need to balance strict virial consistency against computational efficiency.

Let us study the complexity of evaluating $\bm{P}^{\mathcal{N}}_{\mathrm{coul}}$ and $\bm{P}^{\mathcal{F}}_{\mathrm{coul}}$. The short-range cost is $O(r_c^3N)$ by truncation. For direct Fourier-space summation in Eq.~\eqref{eq::P_long}, truncating to $|\bm{k}|\le k_c$ gives $k_c = O(1/\widetilde{s}_0)=O(1/r_c)$
from Eq.~\eqref{eq::continuity}. The number of Fourier modes scales as
$O(\det(\bm{h})\,k_c^3)$. At fixed density, $\det(\bm{h})=O(N)$, so the long-range cost is
$O(N^2/r_c^3)$. The total cost is
$O(N r_c^3 + N^2/r_c^3)$, minimized at $r_c=O(N^{1/6})$ with complexity
$O(N^{3/2})$. If FFT-based methods are used, one takes $r_c=O(1)$ and the complexity can be reduced to $O(N\log N)$,
but it requires substantial global communication (3D FFT transposes and reductions over $O(N)$ Fourier modes),
which often limits parallel scalability~\cite{ayala2021scalability}.

\subsection{Random batch sum-of-Gaussians method with measure recalibration} 
\label{sec::Sampling}
We now present a stochastic algorithm for the long-range pressure contribution. It replaces full Fourier summation with random batch importance sampling and further reduces variance by measure recalibration. The resulting RBSOG method approximates $\bm{P}_{\text{coul}}^{\mathcal{F}}$ with $O(N)$ computational cost and $O(1)$ communication cost per evaluation.

We begin by briefly introducing the basic idea of random batch. For a lattice sum $\mu=\sum_{\bm{k}} f(\bm{k})$ on the Fourier lattice
$\bm{k}\in2\pi\bm{h}^{-\top}\mathbb{Z}^3$, we choose a proposal distribution $h(\bm{k})$ and rewrite
\begin{equation}\label{eq::randombatch}
\mu=\sum_{\bm{k}} \frac{f(\bm{k})}{h(\bm{k})}\,h(\bm{k})
=\mathbb{E}_{\bm{k}\sim h(\bm{k})}\!\left[\frac{f(\bm{k})}{h(\bm{k})}\right],
\end{equation}
where $\mathbb{E}_{\bm{k} \sim h(\boldsymbol{k})}$ denotes expectation over $\bm{k}\sim h(\bm{k})$. Instead of evaluating the full sum, we estimate the expectation using a mini-batch of Fourier modes with batch size $P$, yielding an unbiased approximation of $\mu$. The efficiency depends on the importance weights $f(\bm{k})/h(\bm{k})$: choosing $h(\bm{k})$ to match the behavior of $|f(\bm{k})|$ typically reduces variance and allows a smaller batch size. It is worth noting that the random mini-batch strategy was firstly introduced in interacting particle systems by Jin et al.~\cite{jin2020random}, and have been studied in various areas including nonconvex optimization, Monte Carlo simulations, and quantum simulations~\cite{li2020random,cai2024convergence,carrillo2021consensus,ye2021efficient}.

However, a direct application of Eq.~\eqref{eq::randombatch} to Eq.~\eqref{eq::P_long} is inefficient because radially symmetric and non-radially symmetric contributions have different optimal proposals. To clarify this point, we split  $\bm{P}_{\mathrm{coul}}^{\mathcal{F}}=\bm{P}_{\mathrm{coul}}^{\mathcal{F},\mathrm{r}}+\bm{P}_{\mathrm{coul}}^{\mathcal{F},\mathrm{nr}}$, where the radial and non-radial parts are given by
\begin{equation}\label{eq::PcoulFr}
\bm{P}_{\text{coul}}^{\mathcal{F},\mathrm{r}}
:=\frac{1}{4\det(\bm{h})^2}\sum_{\bm{k}\neq \bm{0}}\sum_{\ell=0}^{\widetilde{M}-1}
\pi^{3/2}\widetilde{w}_\ell \widetilde{s}_\ell^5
e^{-\widetilde{s}_\ell^2 k^2/4}\,|\rho(\bm{k})|^2 \,\bm{I}
\end{equation}
and
\begin{equation}\label{eq::PcoulFnr}
\bm{P}_{\text{coul}}^{\mathcal{F},\mathrm{nr}}
:=- \frac{1}{8\det(\bm{h})^2}\sum_{\bm{k}\neq \bm{0}}\sum_{\ell=0}^{\widetilde{M}-1}
\pi^{3/2}\widetilde{w}_\ell \widetilde{s}_\ell^7
e^{-\widetilde{s}_\ell^2 k^2/4}\,|\rho(\bm{k})|^2 \,\bm{k}\otimes\bm{k},
\end{equation}
respectively. To approximate the radial part Eq.~\eqref{eq::PcoulFr}, a natural choice of proposal distribution is 
\begin{equation}\label{eq::samplingP}
\mathscr{P}^{\mathrm{r}}(\bm{k}) =\frac{1}{S^{\mathrm{r}}}\sum_{\ell=0}^{\widetilde{M}-1} \pi^{3/2}\widetilde{w}_\ell \widetilde{s}_\ell^5\,k^2\, e^{-\widetilde{s}_\ell^2 k^2/4},
\end{equation}
where 
\begin{equation}
\label{eq::Sr}
S^{\text{r}}:=\sum_{\ell=0}^{\widetilde{M}-1}H_{\ell}^{\text{r}} \quad \text{with}\quad H_{\ell}^{\text{r}}:=\sum_{\bm{m}\in\mathbb{Z}^3}\pi^{3/2}\widetilde{w}_\ell \widetilde{s}_\ell^5|2\pi\bm{h}^{-\top}\bm{m}|^2e^{-\widetilde{s}_\ell^2 |2\pi\bm{h}^{-\top}\bm{m}|^2/4}
\end{equation}
is the normalization constant. Drawing $P^{\text{r}}$ Fourier modes $\{\bm{k}_{p}^{\text{r}}\}_{p=1}^{P^{\text{r}}}$ from $ \mathscr{P}^{\mathrm{r}}(\bm{k})$, Eq.~\eqref{eq::PcoulFr} can be approximated by
\begin{equation}\label{eq::P_approx_r}
\bm{P}_{\text{coul}}^{\mathcal{F},\mathrm{r}}
\approx
\bm{P}_{\text{coul}}^{\mathcal{F},\mathrm{r}*}
:=\frac{S^{\mathrm{r}}}{4\det(\bm{h})^2P^{\mathrm{r}}}
\sum_{p=1}^{P^{\mathrm{r}}}\frac{|\rho(\bm{k}_p^{\text{r}})|^2}{|\bm{k}_p^{\text{r}}|^2}\,\bm{I}.
\end{equation}
This choice follows from low-$k$ behavior: by charge neutrality, $|\rho(\bm{k})|^2=O(k^2)$ as $k\to 0$, so $|\rho(\bm{k})|^2/|\bm{k}|^2$ varies slowly for long-wave modes. Since low-$k$ modes dominate long-range fluctuations, Eq.~\eqref{eq::P_approx_r} has relatively small variance. For the non-radial term Eq.~\eqref{eq::PcoulFnr}, the analogous low-$\bm{k}$-matched proposal is
\begin{equation}\label{eq::samplingPnr}
\mathscr{P}^{\text{nr}}(\bm{k})=\frac{1}{S^{\text{nr}}}\sum_{\ell=0}^{\widetilde{M}-1}\pi^{3/2}\widetilde{w}_\ell \widetilde{s}_\ell^7 k^4 e^{-\widetilde{s}_\ell^2 k^2/4},
\end{equation}
which differs from $\mathscr{P}^{\mathrm{r}}(\bm{k})$ by an extra factor $\widetilde{s}_\ell^2 k^2$. The normalization factor is
\begin{equation}
\label{eq::Snr}
S^{\text{nr}}:=\sum_{\ell=0}^{\widetilde{M}-1}H_{\ell}^{\text{nr}}, \quad \text{with}\quad H_{\ell}^{\text{nr}}:=\sum_{\bm{m}\in\mathbb{Z}^3}\pi^{3/2}\widetilde{w}_\ell \widetilde{s}_\ell^7|2\pi\bm{h}^{-\top}\bm{m}|^4e^{-\widetilde{s}_\ell^2 |2\pi\bm{h}^{-\top}\bm{m}|^2/4}.
\end{equation}
Drawing $P^{\text{nr}}$ samples $\{\bm{k}_{p}^{\text{nr}}\}_{p=1}^{P^{\text{nr}}}$ from $\mathscr{P}^{\text{nr}}$,  we obtain the estimator
\begin{equation}\label{eq::P_approx_nr}
\bm{P}_{\text{coul}}^{\mathcal{F},\text{nr}}\approx \bm{P}_{\text{coul}}^{\mathcal{F},\text{nr}*}:=-\frac{1}{8\det(\bm{h})^2}\frac{S^{\text{nr}}}{P^{\text{nr}}}\sum_{p=1}^{P^{\text{nr}}}\frac{|\rho(\bm{k}_{p}^{\text{nr}})|^2\bm{k}_{p}^{\text{nr}}\otimes \bm{k}_{p}^{\text{nr}}}{|\bm{k}_{p}^{\text{nr}}|^4}.
\end{equation}

In Eqs.~\eqref{eq::P_approx_r} and \eqref{eq::P_approx_nr}, the computational bottleneck is the evaluation of \(\rho(\bm{k})\), which incurs \(O(N)\) work per mode and requires a global reduction in parallel implementations. If the radial and non-radial contributions are evaluated independently, one needs \(P^{\mathrm r}+P^{\mathrm{nr}}\) mode evaluations per step. Alternatively, adopting a shared proposal (e.g., \(\mathscr{P}^{\mathrm r}(\bm{k})\)) reduces the number of mode evaluations to \(P^{\mathrm r}\). However, the resulting importance weights for the non-radial contribution become highly unbalanced, leading to a noticeable variance inflation. A related effect has been reported for Ewald-based RBE methods~\cite{liang2022superscalability,liang2022iso}, where relatively large batch sizes are required for stable NPT simulations.

To address this issue, we propose a measure-recalibration strategy: radial and non-radial components are treated separately, but only $P^{\mathrm{r}}$ evaluations of $\rho(\bm{k})$ are needed per step. We first draw $\{\bm{k}_{p}^{\mathrm r}\}_{p=1}^{P^{\mathrm r}}$ from $\mathscr{P}^{\mathrm r}(\bm{k})$ using the Metropolis--Hastings (MH) algorithm~\cite{hastings1970monte} (see~\ref{app::MHsamp}), then reuse these samples as proposals in a second MH chain with the target distribution $\mathscr{Q}(\bm{k})=\mathscr{P}^{\mathrm{nr}}(\bm{k})$. At step $t$ of the second MH chain, we propose $\bm{k}^*=\bm{k}_{t}^{\mathrm r}$ with proposal distribution $q(\bm{k}^*|\bm{k}_t)=\mathscr{P}^{\mathrm r}(\bm{k}^*)$, where $\bm{k}_t$ is the current state. The acceptance probability is
\begin{equation}
\text{Acc}(\bm{k}^*,\bm{k}_{t}):=\min\left\{1,\frac{\mathscr{Q}(\bm{k}^*)q(\bm{k}_t|\bm{k}^*)}{\mathscr{Q}(\bm{k}_t)q(\bm{k}^*|\bm{k}_{t})}\right\}=\min\left\{1,\frac{\mathscr{P}^{\mathrm{nr}}(\bm{k}^*)\mathscr{P}^{\mathrm{r}}(\bm{k}_{t})}{\mathscr{P}^{\mathrm{nr}}(\bm{k}_{t})\mathscr{P}^{\mathrm{r}}(\bm{k}^*)}\right\}.
\end{equation}
If accepted, we set $\bm{k}_{t+1}=\bm{k}^*$; otherwise, we keep $\bm{k}_{t+1}=\bm{k}_t$. Because $\mathscr{P}^{\mathrm{nr}}$ and $\mathscr{P}^{\mathrm r}$ are close, acceptance is high (above $75\%$ in our tests). Note that the proposal probability $q(\bm{k}^*|\bm{k}_{t})$ used here does not depend on the current state $\bm{k}_t$. The resulting states $\{\bm{k}_{p}^{\mathrm{nr}}\}_{p=1}^{P^{\mathrm r}}$ are  used in Eq.~\eqref{eq::P_approx_nr}. The number of \emph{distinct} modes in $\{\bm{k}_{p}^{\mathrm r}\}\cup\{\bm{k}_{p}^{\mathrm{nr}}\}$ is at most $P^{\mathrm r}$, so at most $P^{\mathrm r}$ structure-factor evaluations are required per step. For convenience, we denote $P^{\text{r}}=P^{\text{nr}}=P$ in the remainder of this work. 

The normalization constants $S^{\mathrm r}$ and $S^{\mathrm{nr}}$ are lattice sums over $\mathbb{Z}^3$. They can be computed by truncated direct sums or via alternative representations derived from the Poisson summation formula (see Proposition~\ref{prop::Poisson}). The choice between Eq.~\eqref{eq::Sr} and Eq.~\eqref{eq::H2} (and likewise between Eq.~\eqref{eq::Snr} and Eq.~\eqref{eq::H2r}) is guided by whether \(\|\bm{h}\|/s_{\ell}\) or \(s_{\ell}\|\bm{h}^{-\top}\|/2\) exceeds $1$. In practice, truncation to $|\bm{m}|\leq \min \{\widetilde{s}_{\ell}\sqrt{\log(1/\varepsilon)}/\|\bm{h}\|,\,2\sqrt{\log(1/\varepsilon)}\|\bm{h}\|/\widetilde{s}_{\ell}\}$ is sufficient.
\begin{proposition}\label{prop::Poisson} 
By applying the Poisson summation formula to Eq.~\eqref{eq::Sr} and Eq.~\eqref{eq::Snr}, we have \begin{equation}\label{eq::H2} H_{\ell}^{\mathrm{r}}=\det(\bm{h})\widetilde{w}_{\ell}\left[6A_{0,\ell}-\frac{4}{\widetilde{s}_{\ell}^2}A_{2,\ell}\right] \end{equation} 
and
\begin{equation}\label{eq::H2r} H_{\ell}^{\mathrm{nr}}=\det(\bm{h}) \widetilde{w}_{\ell}\left[60A_{0,\ell}-\frac{80}{\widetilde{s}_{\ell}^2}A_{2,\ell}+\frac{16}{\widetilde{s}_{\ell}^4}A_{4,\ell}\right], \end{equation} where \begin{equation} A_{0,\ell}:=\sum_{\bm{m}\in\mathbb{Z}^3}e^{-\frac{|\bm{h}\bm{m}|^2}{\widetilde{s}_{\ell}^2}},\quad A_{2,\ell}:=\sum_{\bm{m}\in\mathbb{Z}^3}|\bm{h}\bm{m}|^2e^{-\frac{|\bm{h}\bm{m}|^2}{\widetilde{s}_{\ell}^2}},\quad A_{4,\ell}:=\sum_{\bm{m}\in\mathbb{Z}^3}|\bm{h}\bm{m}|^4e^{-\frac{|\bm{h}\bm{m}|^2}{\widetilde{s}_{\ell}^2}}. \end{equation} 
\end{proposition} 

In MD simulations, we use the stochastic pressure tensor 
\begin{equation}
\bm{P}_{\text{coul}}^{\mathcal{F}*}:=\bm{P}_{\text{coul}}^{\mathcal{F},\mathrm{r}*}+\bm{P}_{\text{coul}}^{\mathcal{F},\text{nr}*}
\end{equation}
to approximate $\bm{P}_{\text{coul}}^{\mathcal{F}}$. Since the number of distinct Fourier modes is at most $P$, evaluating $\bm{P}_{\text{coul}}^{\mathcal{F}*}$ costs $O(PN)$ computation and $O(P)$ communication per time step. This implies that the RBSOG method has linear complexity when $P=O(1)$. Section~\ref{sec::Error} shows that $\bm{P}_{\text{coul}}^{\mathcal{F}*}$ is unbiased with bounded variance. In the NPT updates, we also use the random batch idea in Eq.~\eqref{eq::randombatch} to construct the force estimator (see~\ref{app::RBForce}). The overall RBSOG method is summarized in Algorithm~\ref{alg::RBSOG-NPT}.

The proposed RBSOG method offers several advantages. Compared with FFT-based methods~\cite{Darden1993JCP,di2015stochastic,lindbo2011spectral}, RBSOG reduces the communication cost by replacing FFT with the random batch importance sampling in Fourier space, and the computational cost is reduced from $O(N\log N)$ to $O(N)$. Compared with Ewald-based RBE method~\cite{liang2022iso}, RBSOG replaces the classical Ewald decomposition with the SOG decomposition, addressing the discontinuity issue and reducing truncation error. Moreover, the measure-recalibration strategy further reduces variance at the same batch size. In the numerical tests of Section~\ref{sec::NumResult}, we observe that RBSOG requires substantially fewer sampled modes (often $1/4$ or less) to reach the same accuracy as the RBE baseline. Finally, by combining with kernel-independent SOG methods~\cite{greengard2018anisotropic,gao2022kernel,lin2025weighted}, the RBSOG framework can be readily extended to other long-range kernels.

\begin{algorithm}[H]
\caption{Random batch sum-of-Gaussians method for NPT simulations}\label{alg::RBSOG-NPT}
\begin{algorithmic}[1]
\Require Select cutoff $r_c$; time step $\Delta t$; total simulation steps $N_{\text{step}}$; batch size $P$; SOG parameters $\{\widetilde{w}_\ell,\widetilde{s}_\ell\}_{\ell=0}^{\widetilde{M}-1}$; initial positions $\{\bm{r}_i\}$ and velocities $\{\bm{v}_i\}$; initial cell tensor $\bm{h}$.
\For{$n=1,2,\ldots,N_{\text{step}}$}
\State Sample $P$ modes $\{\bm{k}_p^{\mathrm r}\}_{p=1}^{P} \sim \mathscr{P}^{\text{r}}(\bm{k})$ using MH (\ref{app::MHsamp}).
\State Sample $\{\bm{k}_p^{\mathrm{nr}}\}_{p=1}^{P}\sim\mathscr{P}^{\text{nr}}(\bm{k})$ by the measure-recalibration strategy in Section~\ref{sec::Sampling}, reusing $\{\bm{k}_p^{\mathrm r}\}$ as proposals.
\State Compute the short-range pressure $\bm{P}_{\mathrm{coul}}^{\mathcal{N}}$ (Eq.~\eqref{eq::P_short}) and short-range force $\bm{F}_{i,\text{coul}}^{\mathcal{N}}$ (Eq.~\eqref{eq::force_SOG}) by direct truncation at the cutoff $r_c$.
\State For each distinct mode in $\{\bm{k}_p^{\mathrm r}\}\cup\{\bm{k}_p^{\mathrm{nr}}\}$, evaluate and reduce the structure factor $\rho(\bm{k})$.
\State Compute the radial and non-radial long-range pressure estimators $\bm{P}_{\mathrm{coul}}^{\mathcal{F},\mathrm{r}*}$ (Eq.~\eqref{eq::P_approx_r}) and $\bm{P}_{\mathrm{coul}}^{\mathcal{F},\mathrm{nr}*}$ (Eq.~\eqref{eq::P_approx_nr}), and compute the long-range force estimator $\bm{F}_{i,\mathrm{coul}}^{\mathcal{F}*}$ (Eq.~\eqref{eq::f_approx}).
\State Form the total Coulombic pressure and force estimators, $\bm{P}_{\mathrm{coul}}^{*}=\bm{P}_{\mathrm{coul}}^{\mathcal{N}}+\bm{P}_{\mathrm{coul}}^{\mathcal{F},\mathrm{r}*}+\bm{P}_{\mathrm{coul}}^{\mathcal{F},\mathrm{nr}*}$ and $\bm{F}_{i,\mathrm{coul}}^{*}=\bm{F}_{i,\mathrm{coul}}^{\mathcal{N}}+\bm{F}_{i,\mathrm{coul}}^{\mathcal{F}*}$. 
\State Integrate the NPT equations of motion with the chosen thermostat/barostat using $\bm{P}_{\mathrm{coul}}^{*}$ and $\bm{F}_{i,\mathrm{coul}}^{*}$.
\EndFor
\Ensure trajectories $\{(\{\bm{r}_i^{n}\},\{\bm{v}_i^{n}\},\bm{h}^{n})\}_{n=1}^{N_{\text{step}}}$.
\end{algorithmic}
\end{algorithm}

\subsection{Error Estimates and Parameter Selection for SOG Decomposition}
\label{sec::Error}
In this section, we extend the pointwise SOG error estimate in Eq.~\eqref{eq::approxerror} to pressure-tensor errors. The decomposition error of the pressure tensor can be written as
\begin{equation}
   \label{eq::press_error}
    \begin{aligned}
    \bm{P}_{\text{coul},\text{err}}=\frac{1}{2\det(\bm{h})}\sum_{\bm{n}\in\mathbb{Z}^3}\sum_{i,j=1}^{N}q_iq_j K(|\bm{r}_{ij}+\bm{n}\bm{h}|)(\bm{r}_{ij}+\bm{n}\bm{h})\otimes(\bm{r}_{ij}+\bm{n}\bm{h}),
    \end{aligned}
\end{equation}
where the error kernel is defined by
\begin{equation}\label{eq::ErrKernel}
K(r):=\left(\frac{1}{r^3}-\sum_{\ell=0}^{\widetilde{M}-1}\widetilde{w}_\ell e^{-r^2/\widetilde{s}_\ell^2}\right) H(r-r_c).
\end{equation}
Here $H(r)$ is the Heaviside step function with $H(r)=1$ for $r\ge 0$ and $0$ otherwise. Applying the Poisson summation formula to Eq.~\eqref{eq::press_error} gives
\begin{equation}\label{eq::P_err_Fourier}
\bm{P}_{\mathrm{coul},\mathrm{err}}=
\frac{1}{2\det(\bm{h})^2}\sum_{\bm{k}}|\rho(\bm{k})|^2\widehat{\widetilde{\bm K}}(\bm{k}),
\end{equation}
where $\widehat{\widetilde{\bm K}}(\bm{k})$ denotes the Fourier transform of $K(r)\bm{r}\otimes\bm{r}$.
We next state Lemma~\ref{lem::Fourierexpansion}, which gives the explicit form of $\widehat{\widetilde{\bm K}}(\bm{k})$, and then derive Theorem~\ref{thm::PressError}.

\begin{lemma}\label{lem::Fourierexpansion}
Let $\widehat{K}(k)$ denote the Fourier transform of $K(r)$. Then $\widehat{\widetilde{\bm K}}(\bm{k})$ can be expressed as
\begin{equation}\label{eq::FourierTransform_Ktilde}
\widehat{\widetilde{\bm K}}(\bm{k})=
-\frac{\widehat{ K}'(k)}{k}\,\bm I
-\left(\widehat{ K}''(k)-\frac{\widehat{ K}'(k)}{k}\right)\frac{\bm k\otimes\bm k}{k^2},
\end{equation}
where $\widehat{K}'(k)$ and $\widehat{K}''(k)$ are the first and second derivatives of $\widehat{K}(k)$ with respect to $k$.
\end{lemma}

\begin{theorem}
\label{thm::PressError}
For the $C^0$-continuous SOG decomposition of the pressure tensor constructed in Section~\ref{sec::Pressure}, the decomposition error can be estimated by
\begin{equation}
\label{eq::Press_C0}
\left\|\bm{P}_{\mathrm{coul},\mathrm{err}}\right\| \simeq O\left(\frac{1}{\det(\bm{h})}\left[(\log \widetilde{b})^{-3/2} e^{-\frac{\pi^2}{2\log \widetilde{b}}}+\widetilde{b}^{-3 \widetilde{M}}+\left(\widetilde{w}_{-1} e^{-r_c^2/ \widetilde{s}_{-1}^2} - (\widetilde{\omega} - 1)\widetilde{w}_0 e^{-r_c^2/ \widetilde{s}_{0}^2}\right)\right]\right),
\end{equation}
where $\simeq$ indicates ``asymptotically equal'' as $\widetilde{b}\rightarrow 1$.
\end{theorem}
\begin{proof}
We split the error kernel $K(\bm{r})$ into three parts,
\begin{equation}
K(r)=T(r) + G_{\text{up}}(r) +  G_{\text{down}}(r),
\end{equation}
where
\begin{equation}
    \label{eq::T_Gup}
    T(r):=\left(\frac{1}{r^3}-\sum_{\ell=-\infty}^{+\infty}\widetilde{w}_\ell e^{-r^2/\widetilde{s}_\ell^2}\right) H(r-r_c),\quad  G_{\text{up}}(r):= \sum_{\ell=\widetilde{M}}^{+\infty}\widetilde{w}_\ell e^{-r^2/\widetilde{s}_\ell^2}H(r-r_c),
\end{equation}
and
\begin{equation}
    \label{eq::Gdown}
    G_{\text{down}}(r):= \left(\sum_{\ell=-\infty}^{-1}\widetilde{w}_\ell e^{-r^2/\widetilde{s}_\ell^2}-(\widetilde{\omega}-1)\widetilde{w}_0e^{-r^2/\widetilde{s}_0^2}\right)H(r-r_c).
\end{equation}
Here, the first term $T(r)$ is associated with trapezoidal quadrature error in discretizing Eq.~\eqref{eq::BSA_beta}; the last two terms arise from truncation of the Gaussian series; and the prefactor $1-\omega$ associated with $\widetilde{w}_0$ comes from the $C^0$-continuous construction in Eq.~\eqref{omega0}. Therefore, $\bm{P}_{\text{coul},\text{err}}$ can be rewritten as
\begin{equation}\label{eq::SOGSplitting}
\bm{P}_{\text{coul},\text{err}}=\frac{1}{2\det(\bm{h})^2}\sum_{\bm{k}}|\rho(\bm{k})|^2\left(\widehat{\widetilde{\bm T}}(\bm{k})+\widehat{\widetilde{\bm G}}_{\text{up}}(\bm{k})+\widehat{\widetilde{\bm G}}_{\text{down}}(\bm{k})\right),
\end{equation}
where $\widehat{\widetilde{\bm T}}(\bm{k})$, $\widehat{\widetilde{\bm G}}_{\text{up}}(\bm{k})$, and $\widehat{\widetilde{\bm G}}_{\text{down}}(\bm{k})$ are the Fourier transforms of $T(r)\bm{r}\otimes\bm{r}$, $G_{\text{up}}(r)\bm{r}\otimes\bm{r}$, and $G_{\text{down}}(r)\bm{r}\otimes\bm{r}$, respectively. For each component, the sum over Fourier modes $\bm{k}$ can be safely approximate by an integral~\cite{kolafa1992cutoff,LIANG2025101759}
\begin{equation}\label{eq::integral}
\sum_{\bm{k}} \simeq \frac{|\det(\bm{h})|}{(2 \pi)^3} \int_0^{\infty} k^2 d k \int_{-1}^1 d \cos \theta \int_0^{2 \pi} d \varphi,
\end{equation}
where $(k, \theta, \varphi)$ are spherical coordinates and $\simeq$ indicates asymptotically equal in the mean-field limit. Applying Eq.~\eqref{eq::integral} to each component of Eq.~\eqref{eq::SOGSplitting}, together with Lemma~\ref{lem::Fourierexpansion}, yields three oscillatory Fourier integrals, which can be estimated by extending Theorem 2 and Appendix C in Ref.~\cite{LIANG2025101759}. The estimates of $\widehat{\widetilde{\bm T}}$, $\widehat{\widetilde{\bm G}}_{\text{up}}$, and $\widehat{\widetilde{\bm G}}_{\text{down}}$ correspond to the three terms in Eq.~\eqref{eq::Press_C0}, respectively.
\end{proof}

Theorem~\ref{thm::PressError} shows that the pressure-tensor error decays as $O(\widetilde{b}^{-3\widetilde{M}})$ until limited by the remaining terms in Eq.~\eqref{eq::Press_C0}.
To validate Theorem~\ref{thm::PressError}, we test an SPC/E bulk water system~\cite{berendsen1987missing} ($21624$ atoms) with cutoff $r_c = 9 \,\text{\AA}$ and define the relative error of diagonal and off-diagonal pressure components by
\begin{equation}
    \label{eq::err_press_r}
    \mathcal{E}_{P}^{\text{diag}}=\max_{1\leq\alpha\leq 3}\left|\frac{(\bm{P}_{\text{RBSOG}})_{\alpha\alpha}-(\bm{P}_{\text{ref}})_{\alpha\alpha}}{(\bm{P}_{\text{ref}})_{\alpha\alpha}}\right|
\end{equation}
and
\begin{equation}
    \label{eq::err_press_nr}
    \mathcal{E}_{P}^{\text{off-diag}}=\max_{1\leq\alpha<\beta\leq 3}\left|\frac{(\bm{P}_{\text{RBSOG}})_{\alpha\beta}-(\bm{P}_{\text{ref}})_{\alpha\beta}}{(\bm{P}_{\text{ref}})_{\alpha\beta}}\right|,
\end{equation}
where the reference solution $P_{\text{ref}}$ is computed to machine precision by using the SOG decomposition with $\widetilde{b}=1.14$ and sufficiently large $\widetilde{M}$. Figure~\ref{fig::press_M} shows the dependence of error on $\widetilde{M}$.
Both diagonal and off-diagonal components exhibit the expected $O(\widetilde{b}^{-3\widetilde{M}})$ decay, in agreement with our theoretical prediction. Table~\ref{table::para_error} then provides recommended SOG parameter sets for practical accuracy targets.

\begin{figure}[ht]
    \centering
    \includegraphics[width=1.0\linewidth]{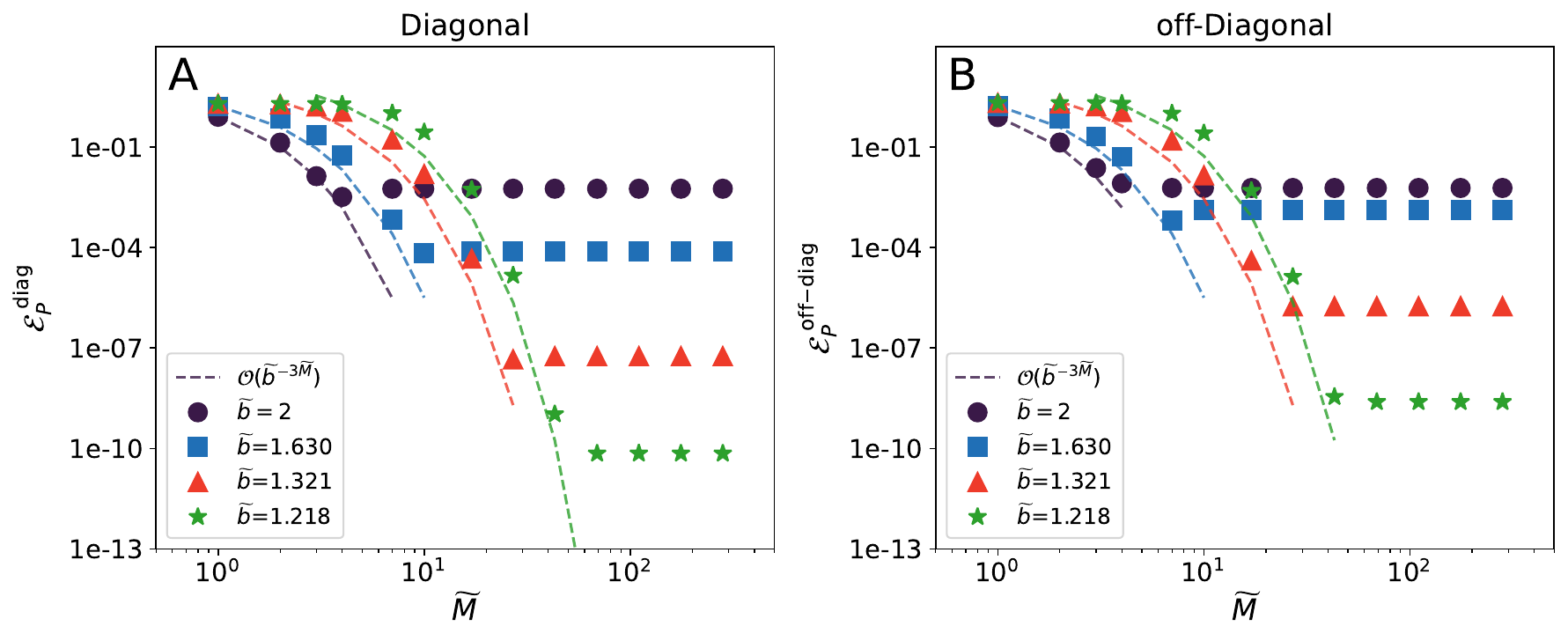}
    \caption{Relative errors of the SOG decomposition introduced in Section~\ref{sec::Pressure} for an SPC/E bulk water system. Panels (A) and (B) show the relative errors of the diagonal and off-diagonal components of the instantaneous pressure tensor, respectively, as functions of the number of Gaussian terms $\widetilde{M}$. The dashed lines in (A)-(B) indicate a fitted decay rate of $\mathcal{O}(\widetilde{b}^{-3\widetilde{M}})$.}
    \label{fig::press_M}
\end{figure}

\renewcommand\arraystretch{1.4}
\begin{table}[!htbp]
\caption{Parameter sets and corresponding errors for the SOG decomposition at $r_c = 9\,\text{\AA}$.
The pressure error is defined as $\max\{\mathcal{E}_{P}^{\text{diag}},\mathcal{E}_{P}^{\text{off-diag}}\}$ and is computed using the pressure decomposition in Section~\ref{sec::Pressure}.
The energy/force error is computed using the decomposition in Section~\ref{sec::useries} with the same splitting parameters, $b=\widetilde{b}$, $\sigma=\widetilde{\sigma}$, and $\omega=\widetilde{\omega}$.
Here $M$ and $\widetilde{M}$ denote the minimal numbers of Gaussian terms required to achieve the listed errors in energy/force and instantaneous pressure tensor, respectively.}
\centering
\vspace{1em}
\scalebox{0.73}{
\begin{tabular}{c|c|c|cc|cc|cc}
\hline \multirow{2}{*}{\large $\widetilde{b}$}& \multirow{2}{*}{\large $\widetilde{\sigma}$}& \multirow{2}{*}{\large $\widetilde{\omega}$} & \multicolumn{2}{c|}{\large Energy} & \multicolumn{2}{c|}{\large Force} & \multicolumn{2}{c}{\large Pressure}\\
\cline{4-5} \cline{6-7} \cline{8-9}& & &  Error & $M$ &  Error& $M$ &  Error& $\widetilde{M}$\\

\hline 2 & $4.524309831775139$ &$0.994446492762232$ & $1.22\hbox{E}-4$ & $12$ & $3.48\hbox{E}-3$ & $4$ & $3.24\hbox{E}-3$ & $4$\\

\hline $1.62976708826776469$ & $3.270345668108471$ &$1.007806979343806$  & $1.06\hbox{E}-5$ & $25$ & $1.83\hbox{E}-4$ & $7$ &$7.54\hbox{E}-5$ & $10$\\

\hline $1.32070036405934420$ & $2.049434420796892$ &$1.001889141148119$ & $7.80\hbox{E}-9$ & $64$ & $3.27\hbox{E}-7$ & $20$ & $4.59\hbox{E}-8$ & $27$\\

\hline $1.21812525709410644$ & $1.597010732309955$ &$1.000901461560333$ & $1.07\hbox{E}-11$ & $124$ & $6.79\hbox{E}-10$ & $40$ & $7.10\hbox{E}-11$ & $69$\\
\hline
\end{tabular}}
\label{table::para_error}
\end{table}

\begin{remark}
    SOG parameters should be selected according to the quantities that govern the dynamics.
    In NPT simulations, both pressure tensor and force enter the equations of motion, and matching a given pressure tolerance typically requires slightly more SOG terms than matching force tolerance; therefore, we recommend choosing the number of Gaussian terms from the pressure errors in Table~\ref{table::para_error}. In NVT simulations, force accuracy is the primary criterion. The energy accuracy are also useful for applications like Monte Carlo simulations.
\end{remark}

\subsection{Analysis of the RBSOG-based MD}
We next analyze convergence and per-step complexity of RBSOG-based NPT simulations. Define the long-range pressure fluctuation 
\begin{equation}
\label{eq::deviation_press}
\widetilde{\bm{\chi}}:=\bm{P}_{\text{coul}}^{\mathcal{F}}-\bm{P}_{\text{coul}}^{\mathcal{F},\text{r}*}-\bm{P}_{\text{coul}}^{\mathcal{F},\text{nr}*}.
\end{equation}
Its expectation and variance are given in Proposition~\ref{prop::Press_Var}.

\begin{proposition}
    \label{prop::Press_Var}
    Each component of the fluctuation, $\widetilde{\bm{\chi}}_{\mu\nu}$, has zero expectation, and its variance is given by
    \begin{equation}
      \label{eq::Press_var}
      \begin{aligned}
\mathbb{E}|\widetilde{\bm{\chi}}_{\mu\nu}|^2\le&\frac{\delta_{\mu\nu}}{P}\left[\frac{\pi^{3/2} S^{\text{r}}}{|\det(\bm{h})|^4}\sum_{\boldsymbol{k} \neq 0} \left(\sum_{\ell=0}^{\widetilde{M}-1}\widetilde{w}_\ell \widetilde{s}_\ell^5e^{-\widetilde{s}_\ell^2|\bm{k}|^2/4}\right)\frac{|\rho(\bm{k})|^4}{|\bm{k}|^2}-|(\bm{P}_{\mathrm{coul}}^{\mathcal{F},\mathrm{r}})_{\mu\nu}|^2\right]\\
&+\frac{1}{P}\left[\frac{\pi^{3/2}S^{\mathrm{nr}}}{36|\det(\bm{h})|^4}\sum_{\boldsymbol{k} \neq 0} \left(\sum_{\ell=0}^{\widetilde{M}-1}\widetilde{w}_\ell \widetilde{s}_\ell^7e^{-\widetilde{s}_\ell^2|\bm{k}|^2/4}\right)\frac{|\rho(\bm{k})|^4k_\mu^2k_\nu^2}{|\bm{k}|^4}-|(\bm{P}_{\mathrm{coul}}^{\mathcal{F},\mathrm{nr}})_{\mu\nu}|^2\right],
      \end{aligned}
    \end{equation}
where $\delta_{\mu\nu}=1$ if $\mu=\nu$ and $\delta_{\mu\nu}=0$ otherwise, and $\mu,\nu\in\{1,2,3\}$ denote Cartesian components.
\end{proposition}

Proposition~\ref{prop::Press_Var} implies that the RBSOG approximation of the long-range pressure tensor is unbiased, i.e., $\mathbb{E}\left[\bm{P}_{\text{coul}}^{\mathcal{F},\text{r}*}+\bm{P}_{\text{coul}}^{\mathcal{F},\text{nr}*}\right]=\bm{P}_{\text{coul}}^{\mathcal{F}}$. Moreover, Eq.~\eqref{eq::Press_var} shows that the variance of each component scales as $O(1/P)$. A sharper statement is given by Theorem~\ref{thm::conv_rate_P} under Debye--H\"{u}ckel (DH) theory (see~\cite{levin2002electrostatic} and~\ref{app::DH_theory}).

\begin{theorem}
    \label{thm::conv_rate_P}
    Let $\eta_{\mathrm{sys}}=N/\det(\bm{h})$ be the number density. Under the DH theory, the variance of the random batch approximation of the pressure tensor, $\mathbb{E}|\widetilde{\bm{\chi}}|^2$, decays as $O(1/P)$ with respect to the batch size $P$, and the bound is independent of the particle number $N$. 
\end{theorem}

\begin{proof}
    By the DH theory in~\ref{app::DH_theory}, we have the bound
    \begin{equation}
        \label{eq::DH_upperbound}
        |\rho(\bm{k})|^4\le CN^2q_{\text{max}}^4,
    \end{equation}
    where $C$ is independent of $N$, and $q_{\text{max}}=\max\limits_{i}\{|q_i|\}$. We write $\widetilde{\bm{\chi}}:=\widetilde{\bm{\chi}}^{\text{r}}+\widetilde{\bm{\chi}}^{\text{nr}}$, with $\widetilde{\bm{\chi}}^{\text{r}}=\bm{P}_{\text{coul}}^{\mathcal{F},\text{r}}-\bm{P}_{\text{coul}}^{\mathcal{F},\text{r}*}$ and $\widetilde{\bm{\chi}}^{\text{nr}}=\bm{P}_{\text{coul}}^{\mathcal{F},\text{nr}}-\bm{P}_{\text{coul}}^{\mathcal{F},\text{nr}*}$. Substituting Eq.~\eqref{eq::DH_upperbound} into Eq.~\eqref{eq::Press_var} and using the integral approximation Eq.~\eqref{eq::integral}, we obtain
    \begin{equation}
        \label{eq::press_var_r}
        \begin{aligned}
\mathbb{E}\left(|\widetilde{\bm{\chi}}_{\mu\nu}^{\text{r}}|^2\right)&\le \frac{C_1}{P}\frac{\pi^{3/2} S^{\text{r}}\eta_{\mathrm{sys}}^2q_{\text{max}}^4}{(2\pi)^3\det(\bm{h})}\int_{0}^{\infty} 4\pi k^2 \left(\sum_{\ell=0}^{\widetilde{M}-1}\widetilde{w}_\ell \widetilde{s}_\ell^5e^{-\widetilde{s}_\ell^2k^2/4}\right)\frac{1}{k^2}\mathrm{d}k\\
&= \frac{C_1}{P}\frac{ S^{\text{r}}\eta_{\mathrm{sys}}^{2}q_{\text{max}}^4}{2\det(\bm{h})} \sum_{\ell=0}^{\widetilde{M}-1}\widetilde{w}_\ell \widetilde{s}_\ell^4.\\
        \end{aligned}
    \end{equation}
where $C_1/C$ is a constant arising from the integral approximation to the series over $\bm{k}$. Next, by Eq.~\eqref{eq::Sr}, the normalization constant $S^{\text{r}}$ satisfies
\begin{equation}
    \label{eq::Sr_calc}
    \begin{aligned}
        S^{\text{r}}\le C_2\sum_{\ell=0}^{\widetilde{M}-1}\widetilde{w}_\ell \widetilde{s}_\ell^5 \frac{\det(\bm{h})}{(2\pi)^3}\int_{0}^{\infty}4\pi k^4 e^{-\widetilde{s}_\ell^2 k^2/4}\mathrm{d}k=6\pi^{-3/2}C_2\det(\bm{h})\sum_{\ell=0}^{\widetilde{M}-1}\widetilde{w}_\ell,
    \end{aligned}
\end{equation}
where $C_2$ is again from the integral approximation. Since both $\sum_{\ell=0}^{\widetilde{M}-1}\widetilde{w}_\ell \widetilde{s}_\ell^4$ and $\sum_{\ell=0}^{\widetilde{M}-1}\widetilde{w}_\ell$ are $O(1)$, we have $S^{\text{r}}=O(\det(\bm{h}))$, and hence $\mathbb{E}\left(|\widetilde{\bm{\chi}}_{\mu\nu}^{\text{r}}|^2\right)=O(1/P)$. Similarly, for the non-radial term, using $k_\mu,k_\nu\le |\bm{k}|$, we have
 \begin{equation}
        \label{eq::press_var_nr}
        \begin{aligned}
\mathbb{E}\left(|\widetilde{\bm{\chi}}_{\mu\nu}^{\text{nr}}|^2\right)&\le \frac{C_1}{P}\frac{\pi^{3/2} S^{\text{nr}}\eta_{\text{sys}}^2q_{\text{max}}^4}{288\pi^3\det(\bm{h})}\int_{0}^{\infty} 4\pi k^2 \left(\sum_{\ell=0}^{\widetilde{M}-1}\widetilde{w}_\ell \widetilde{s}_\ell^7e^{-\widetilde{s}_\ell^2k^2/4}\right)\mathrm{d}k\\
&= \frac{C_1}{P}\frac{ S^{\text{nr}}\eta_{\text{sys}}^2q_{\text{max}}^4}{36\det(\bm{h})} \sum_{\ell=0}^{\widetilde{M}-1}\widetilde{w}_\ell \widetilde{s}_\ell^4,
        \end{aligned}
    \end{equation}
with 
\begin{equation}
    \label{eq::Snr_calc}
    \begin{aligned}
        S^{\text{nr}}
        \le C_2\sum_{\ell=0}^{\widetilde{M}-1}\widetilde{w}_\ell \widetilde{s}_\ell^7 \frac{\det(\bm{h})}{(2\pi)^3}\int_{0}^{\infty}4\pi k^6 e^{-\widetilde{s}_\ell^2 k^2/4}\mathrm{d}k=60\pi^{-3/2}C_2\det(\bm{h})\sum_{\ell=0}^{\widetilde{M}-1}\widetilde{w}_\ell, 
    \end{aligned}
\end{equation}
which yields $\mathbb{E}\left(|\widetilde{\bm{\chi}}_{\mu\nu}^{\text{nr}}|^2\right)= O(1/P)$. This completes the proof.
\end{proof}

We now discuss convergence of RBSOG-based MD in the NPT ensemble. For analysis, we consider the Langevin equations of motion~\cite{liang2022iso}:

\begin{equation}
    \label{eq::Langevin_NPT}
    \begin{cases}
\dot{\boldsymbol{r}}_i=\dfrac{\boldsymbol{p}_i}{m_i}+\dot{\boldsymbol{h}} \boldsymbol{h}^{-1} \boldsymbol{r}_i \\[0.5em]
\dot{\boldsymbol{p}}_i=\bm{F}_i-\boldsymbol{h}^{-\top} \dot{\boldsymbol{h}}^{\top} \boldsymbol{p}_i-\gamma \boldsymbol{p}_i+\sqrt{2 \gamma k_{\mathrm{B}} T m_i} \dot{\boldsymbol{W}}_i \\[0.5em]
\dot{h}_{\mu \nu}=\dfrac{p_{\mu \nu}^h}{M_{\mu \nu}} \\[0.5em]
\dot{p}_{\mu \nu}^h=\operatorname{det}(\boldsymbol{h})\left[\left(\bm{P}_{\text{ins}}-P_{\text{ext}}-\frac{k_{\mathrm{B}} T}{\operatorname{det}(\boldsymbol{h})}\right) \boldsymbol{h}^{-\top}\right]_{\mu \nu}
-\gamma_{\mu \nu} p_{\mu \nu}^h+\sqrt{2 \gamma_{\mu \nu} k_{\mathrm{B}} T M_{\mu \nu}} \dot{W}_{\mu \nu},
\end{cases}
\end{equation}
where $\{m_i\}_{i=1}^{N}$ and $\bm{M}$ are particle and (virtual) box masses, $\gamma$ and $\gamma_{\mu\nu}$ are damping coefficients, $\bm{W}_{i}$ and $W_{\mu\nu}$ are i.i.d. Wiener processes, and $P_{\text{ext}}$ is the external pressure. Let $\bm{X}^{t}=(\{\bm{r}_{i}\}_{i=1}^{N},\{\bm{p}_{i}\}_{i=1}^{N},\bm{h},\bm{p}^{\bm{h}})_{t}$ denote the trajectory generated by Eq.~\eqref{eq::Langevin_NPT} with exact pressure and force, and let $\bm{X}^{t,*}$ denote the trajectory using the random batch approximations. 
Specifically, we use the pressure estimator in Section~\ref{sec::Sampling} and the force estimator in~\ref{app::RBForce}. Assuming that $\{m_{i}\}_{i=1}^{N}$ and $\bm{M}$ are uniformly bounded, we obtain the following theorem.

\begin{theorem}
    \label{thm::config_strong}
    Suppose the exact pressure and forces are bounded and Lipschitz, and the corresponding random batch approximations are unbiased with bounded variance. Under synchronization coupling (the same initial condition and the same Wiener processes), for any $T>0$ there exists $C(T)>0$ such that
    \begin{equation}
        \label{eq::config_diff}
        \left(\mathbb{E} \left[\frac{1}{N} \sum_{i=1}^{N}\left|\bm{X}^{T}-\bm{X}^{T,*}\right|^2\right]\right)^{1 / 2} \le C(T) \sqrt{\Lambda(N) \Delta t},
    \end{equation}
    where $\Lambda(N)$ is the upper bound on the variance of the random approximation. In the mean-field regime, $\Lambda(N)$ is independent of $N$.
\end{theorem}

The proof of Theorem~\ref{thm::config_strong} follows the arguments in~\cite{jin2021convergence,ye2024error} for the original random batch method~\cite{jin2020random}. In our setting, however, Coulomb contributions to pressure and force are singular when $\bm{r}_i=\bm{r}_j$ for some $i \neq j$, which makes a fully rigorous justification difficult. In practice, however, strong short-range repulsion in standard force fields (e.g., Lennard-Jones or Born-Mayer-Huggins terms~\cite{Frenkel2001Understanding,chen2026random}) suppresses such overlap events, so the bound in Eq.~\eqref{eq::config_diff} is expected to hold effectively. For other thermostat/barostat choices (e.g., Nos\'e--Hoover baths~\cite{martyna1994constant,hoover1985canonical}), we refer to~\cite{Jin2020SISC,liang2022iso}.

Finally, we analyze per step complexity of RBSOG. With cutoff $r_c$, the short-range part costs $O(N)$, while the long-range part costs $O(PN)$ under the random batch sampling and measure recalibration in Section~\ref{sec::Sampling}. Hence the overall complexity is $O(N)$ when $P=O(1)$. Numerical results in Section~\ref{sec::NumResult} show that $P\approx 100$ already provides satisfactory accuracy for water and ionic lipid simulations.

\section{Numerical results}
\label{sec::NumResult}

In this section, we present numerical experiments to assess the accuracy and efficiency of the proposed RBSOG method. We consider three representative systems: (i) the SPC/E bulk water system, (ii) a LiTFSI ionic liquid, and (iii) a dipalmitoylphosphatidylcholine (DPPC) membrane system. These examples cover simple molecular liquids, complex ionic solutions, and biologically relevant heterogeneous systems, respectively, and thus provide a broad test of the robustness and applicability of the method.

The RBSOG method is implemented as a self-maintained module in LAMMPS~\cite{thompson2022lammps} (version 21Nov2023). For comparison, we benchmark against the FFT-accelerated PPPM method~\cite{Hockney1988Computer} and another random batch-based RBE method~\cite{liang2022iso}, both implemented in the same package. For clarity, we denote the three methods simply as PPPM, RBE, and RBSOG in the figures. For Ewald splitting used in PPPM and RBE, the decomposition accuracy is set to \(10^{-5}\), and the splitting parameter $\alpha$ is automatically tuned by the native LAMMPS routines once \(r_c\) is specified. For the SOG-splitting-based RBSOG method, we use the second row of Table~\ref{table::para_error} in all tests, which yields a decomposition error comparable to that of Ewald splitting. All experiments are performed on the Siyuan Mark-I cluster at Shanghai Jiao Tong University, which comprises 936 compute nodes; each node is equipped with two Intel Xeon ICX Platinum 8358 CPUs (2.6~GHz, 32 cores per CPU) and 512~GB of memory.

\subsection{Accuracy results for the bulk water system}
\label{sec::Water}
We first validate the method on the SPC/E bulk water system~\cite{berendsen1987missing}. The system contains $24327$ SPC/E water molecules in a cubic box with initial side length $9~nm$. Initial velocities are sampled from a Maxwell distribution. All chemical bonds are constrained using the SHAKE algorithm~\cite{krautler2001fast}. We apply three-dimensional periodic boundary conditions and run NPT simulations at reference temperature $T=298~K$ and set external pressure $P_{\text{ext}}=1$ atm, with time step $\Delta t=0.5~fs$, using PPPM, RBE, and RBSOG. During equilibration ($200~ps$), the short-range Coulomb and Lennard-Jones interactions both use a cutoff of $0.9~nm$.

For structural properties, we examine oxygen-oxygen (O-O) and oxygen-hydrogen (O-H) radial distribution functions (RDFs) using different methods. The results are shown in Figure~\ref{fig:water_rdf}, where PPPM at $10^{-5}$ accuracy is used as the benchmark, RBE uses batch sizes $P=512, 1024$, and RBSOG uses $P=128, 256$. Both RBE and RBSOG accurately reproduce the structural features. However, zoomed views near the first peak in Figure~\ref{fig:water_rdf}(A-B) show that RBSOG with $P=128\ (256)$ achieves accuracy comparable to RBE with $P=512\ (1024)$, indicating an approximately fourfold variance reduction. Notably, in prior NVT results~\cite{liang2023random}, the variance reduction of RBSOG over RBE was only about twofold, highlighting the importance of the measure-recalibration technique in NPT. These findings are also reflected in dynamical properties, shown in Figure~\ref{fig:water_msd}. Figure~\ref{fig:water_msd}(A-B) reports the diffusion coefficient and shear viscosity, respectively. For SPC/E water, RBSOG with $P=128$ is sufficient to reproduce both structural and dynamical properties in NPT, whereas RBE requires $P=512$. This further highlights the variance-reduction advantage of RBSOG over RBE.

\begin{figure}[ht]
    \centering
    \includegraphics[width=1.0\textwidth]{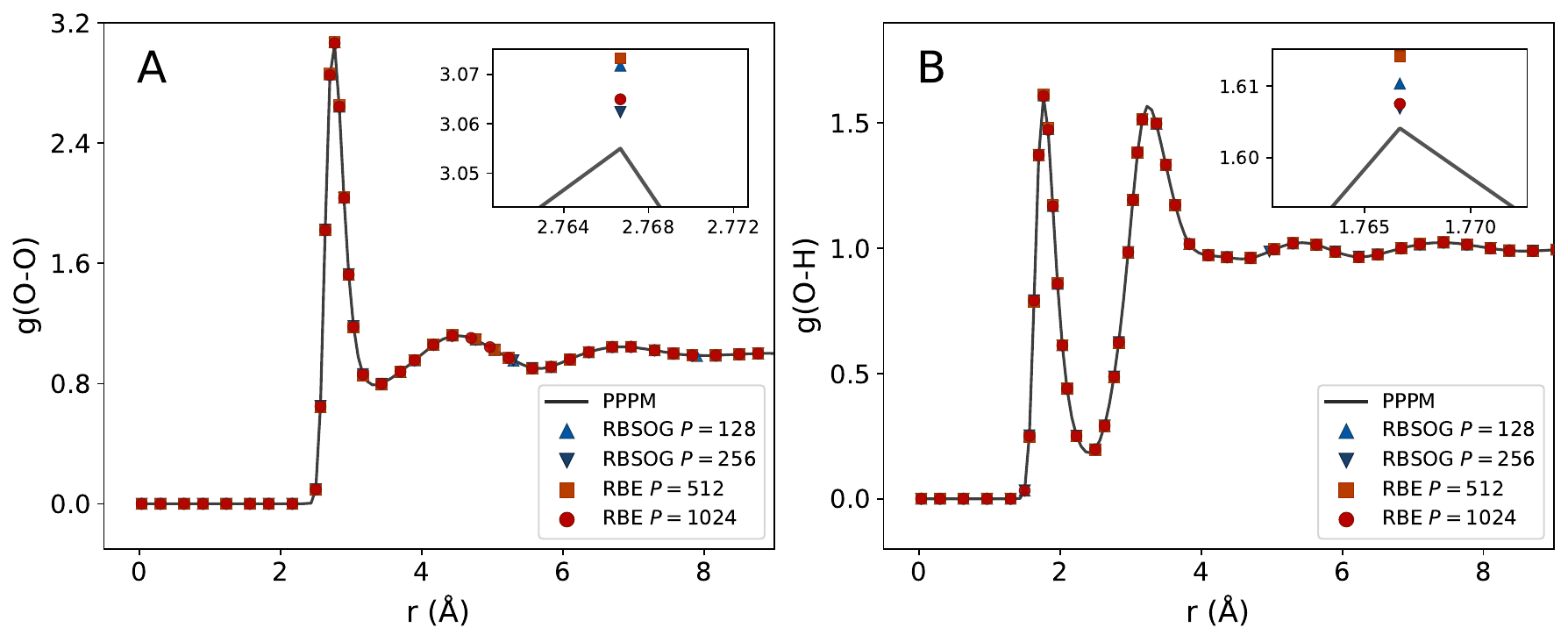}
    \caption{Radial distribution functions of (A) O-O and (B) O-H pairs in the bulk water system. PPPM results are shown as the benchmark (solid line). Results from RBSOG with the batch size $P = 128, \ 256$ and RBE with the batch size $P = 512, \ 1024$ are shown with different markers.}
    \label{fig:water_rdf}
\end{figure}

\begin{figure}[ht]
    \centering
    \includegraphics[width=1.0\textwidth]{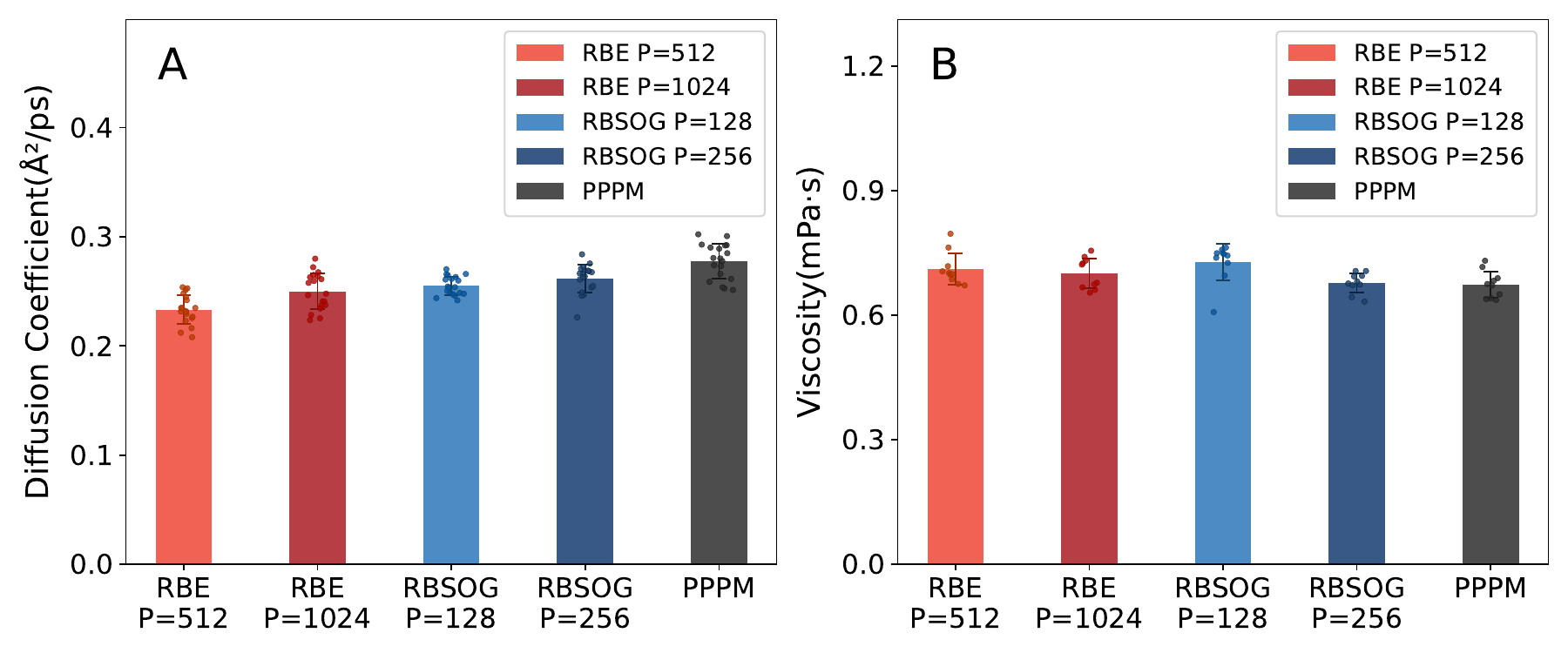}
    \caption{Diffusion coefficient and shear viscosity of bulk water in the NPT ensemble. Results from RBSOG ($P=128,256$) are compared with those from RBE ($P=512,1024$) and PPPM ($\Delta=10^{-5}$). Each quantity is averaged over multiple independent simulations; solid circles indicate different data points, and error bars denote one standard deviation.
    }
    \label{fig:water_msd}
\end{figure}

\subsection{Accuracy results for the LiTFSI ionic liquids}
\label{sec::LiTFSI}
The second accuracy test examines structural properties of LiTFSI ionic liquids at different concentrations. LiTFSI has strong potential in battery-material design because of its favorable low-temperature charge-discharge behavior~\cite{suo2015water}. As concentration varies, ionic solutions undergo fundamental structural changes, including salt-in-water and water-in-salt network states~\cite{malaspina2023unraveling}. RDFs therefore provide a stringent test of algorithmic accuracy for complex systems. We consider LiTFSI concentrations of $1, 5, 10,$ and $21$ mol/L (denoted 1M, 5M, 10M, and 21M). The Li$^+$ and TFSI$^-$ force fields follow Ref.~\cite{Salt_in_water}, while water uses TIP4P/2005~\cite{abascal2005general}. The system contains $8000$ water molecules, and the number of LiTFSI ions is set by concentration. The simulation cell is cubic with side length $7.4677~nm$ under periodic boundary conditions. Simulations are performed in the NPT ensemble at $T=298.0~K$, $P_{\text{ext}}=1$ atm, and time step $\Delta t = 0.5~fs$.

As a benchmark, we first compute lithium-oxygen in water (Li-OW) and lithium-nitrogen (Li-N) RDFs using PPPM at $10^{-5}$ accuracy. The results are shown in Figure~\ref{fig:ionic}. Figure~\ref{fig:ionic}(A) reports reference Li-OW RDFs at different concentrations; peak positions and amplitudes reflect ion-solvation strength and hydration-shell organization. Figure~\ref{fig:ionic}(B) reports reference Li-N RDFs, whose peak features characterize cation-anion interaction range and strength. For each concentration, we then compute the same RDFs with RBE and RBSOG using the same batch-size settings as in Section~\ref{sec::Water}. Figure~\ref{fig:ionic}(C-J) shows that RBSOG reproduces structural properties across concentrations with relatively small sample sizes. Compared with RBE, RBSOG exhibits substantially lower variance in these ionic systems, with a reduction larger than the fourfold improvement observed in bulk water. For example, at 10M, RBSOG with $P=128$ achieves comparable or better RDF accuracy than RBE with $P=1024$. These results highlight the advantage of RBSOG for challenging electrolyte simulations.

\begin{figure}[ht]
    \centering
    \includegraphics[width=0.9\textwidth]{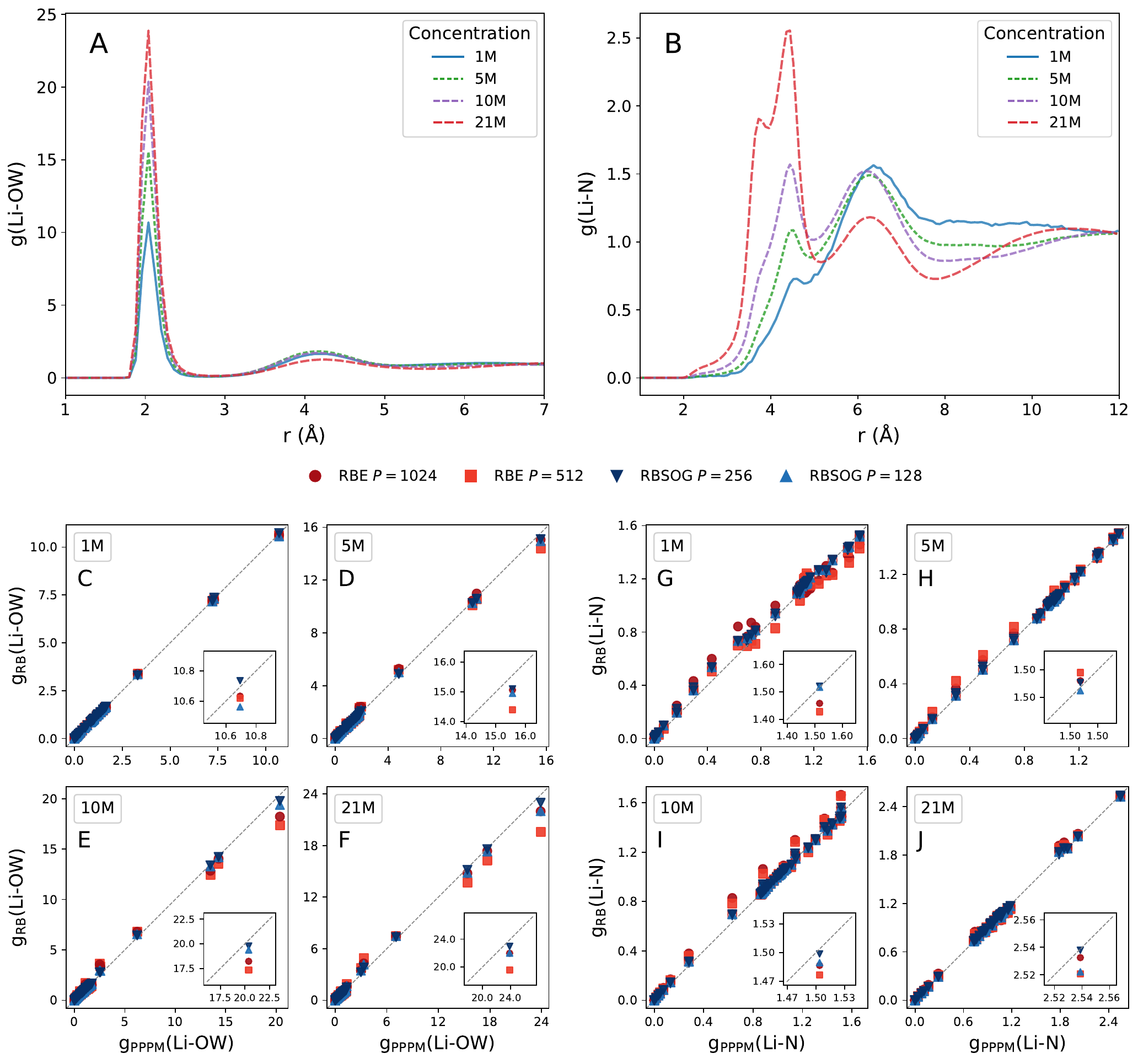}
    \caption{Radial distribution functions (RDFs) of Li-OW (water oxygen) and Li-N pairs in LiTFSI solutions at different concentrations. Panels (A) and (B) show the PPPM reference results computed with $\Delta=10^{-5}$. Panels (C)-(J) compare the results of random batch (RB)-type methods relative to PPPM: Li-OW RDFs in (C)-(F) and Li-N RDFs in (G)-(J), with RBSOG using $P=128,256$ and RBE using $P=512,1024$. Insets show magnified views of the peak regions.}
    \label{fig:ionic}
\end{figure}

\subsection{Accuracy results for the membrane systems}

Membrane systems represented by DPPC play important roles in biology. Simulating such flexible structures in the NPT ensemble is more challenging than the previous examples. In this test, the DPPC membrane is built with CHARMM-GUI~\cite{jo2008charmm} using the CHARMM36m force field~\cite{huang2017charmm36m}. After solvation, the system contains 318 DPPC lipids and 32295 SPC water molecules, for a total of 138225 atoms~\cite{mark2001structure}. The membrane is aligned in the $X$-$Y$ plane with the bilayer center at $Z=0$. Because of membrane geometry, we use semi-isotropic pressure coupling with coupled $X$- and $Y$-direction pressures. Equilibration is performed with the ``fix npt'' command in LAMMPS. We simulate the system with Eq.~\eqref{eq::Langevin_NPT} and the integrator developed in~\cite{liang2022iso}, using PPPM, RBE, and RBSOG, respectively. The reference temperature and pressure are set to $T=303.15~K$ and $P_{\text{ext}}=10^{-3}~Kbar$, with thermostat relaxation time $\gamma = 0.1~ps$, barostat relaxation time $\gamma = 0.5~ps$, and compressibility $4.5 \times 10^{-2}~Kbar^{-1}$. Each simulation runs for $10~ns$, and configurations are saved every 10000 steps ($10~ps$) for statistical analysis. For accurate membrane calculations, the real-space cutoff is set to $1.2~nm$, slightly larger than in the bulk-water and LiTFSI systems. 

For the DPPC bilayer system, membrane thickness and bilayer area are two primary validation metrics. Figure~\ref{fig:mem} shows the time evolution of these quantities, and the corresponding statistics are summarized in Table~\ref{tab:mem}. We compare RBSOG ($P=256$), RBE ($P=256$ and $1024$), and PPPM (tolerance $\Delta=10^{-5}$). Although both RBSOG and RBE qualitatively capture the membrane fluctuations observed with PPPM, RBE at $P=256$ exhibits noticeable deviations from the PPPM benchmark and larger fluctuations. In contrast, RBSOG at $P=256$ and RBE at $P=1024$ show much closer agreement with PPPM, with deviations in the mean membrane thickness and bilayer area below $1\%$. These results further demonstrate the variance-reduction advantage of RBSOG: RBSOG at $P=256$ achieves accuracy comparable to RBE at $P=1024$. Because the cost of the long-range component scales with batch size, this indicates that RBSOG can simulate complex membrane systems with lower computational load and reduced communication overhead.

\begin{figure}[ht]
    \centering
    \includegraphics[width=0.9\textwidth]{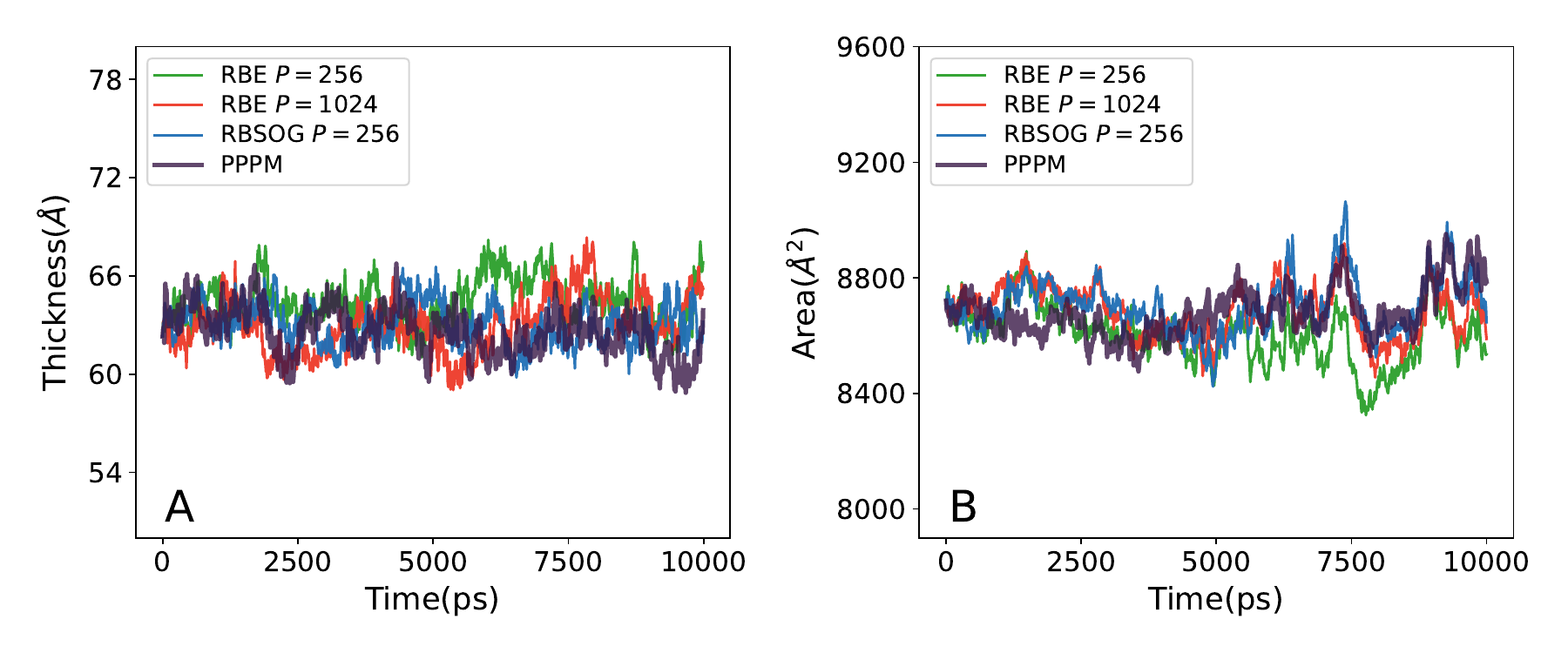}
    \caption{Time evolution of (A) the interbilayer thickness and (B) the area of the simulation box projected onto the $X$--$Y$ plane in DPPC membrane simulations. Data are shown for RBSOG with batch size $P=256$, compared with RBE at $P=256,1024$ and the reference PPPM results.}
    \label{fig:mem}
\end{figure}

\begin{table}[htbp]
\centering
\caption{Mean values and standard deviations of the interbilayer thickness and bilayer area projected onto the $X$--$Y$ plane for the DPPC membrane system. Statistics are computed from the simulation trajectories shown in Fig.~\ref{fig:mem}.}
\vspace{0.5em}
\resizebox{0.65\textwidth}{!}{ 
\begin{tabular}{c|c|cc|c|c} 
\hline
\multicolumn{2}{c|}{} & \multicolumn{2}{c|}{RBE} & RBSOG & PPPM \\
\hline
\multicolumn{2}{c|}{Batch size $P$} & 1024 & 256 & 256 & -- \\
\hline
\multirow{2}{*}{Thickness} & Mean (\AA)  & 63.009 & 64.390 & 63.029 & 62.619 \\
                           & Std. (\AA)   & 1.674  & 1.742  & 1.372  & 1.459 \\
\hline
\multirow{2}{*}{Area}      & Mean (\AA$^2$) & 8685.606 & 8603.001 & 8693.687 & 8678.882 \\
                           & Std. (\AA$^2$)  & 88.227   & 97.916   & 91.328   & 85.050 \\
\hline
\end{tabular}
}
\label{tab:mem}
\end{table}

\subsection{Time performance of the RBSOG-NPT method}
\label{sec::time}

In this subsection, we evaluate the computational and parallel efficiency of the proposed RBSOG method. Performance tests use the SPC/E bulk-water system described in Section~\ref{sec::Water}. To compare PPPM, RBE, and RBSOG at matched accuracy, we use batch sizes $P=512$ for RBE and $P=128$ for RBSOG, while PPPM is run at precision $10^{-5}$. We use elapsed wall-clock time as the performance metric, and all reported values are averaged over $1000$ simulation steps. We first validate linear complexity with respect to system size $N$. With the number of nodes fixed at $\zeta=1$ (64 CPU cores), Figure~\ref{fig:linear} reports computation times for SPC/E systems of varying sizes. The particle number is increased from about $2.2 \times 10^4$ to $1.4 \times 10^6$. Fitted lines on the log-log plot and the dashed $\mathcal{O}(N)$ reference confirm the expected linear scaling of RBSOG. 

\begin{figure}[ht]
    \centering  \includegraphics[width=0.55\textwidth]{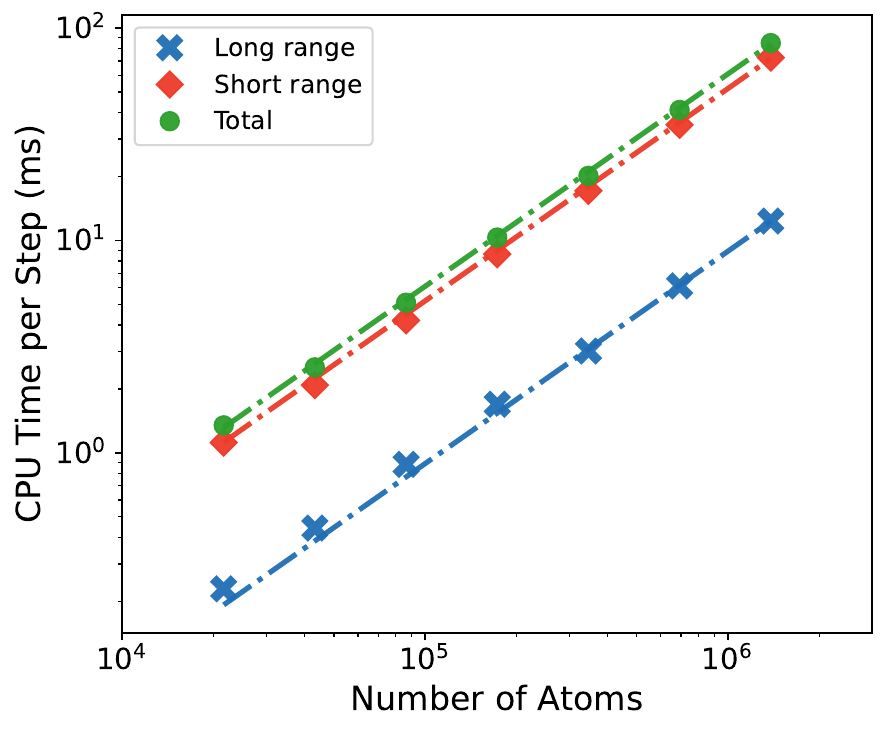}
    \caption{CPU time per step for long-range (Kspace), short-range (Pair), and total (Total) electrostatic calculations as a function of particle number, with the number of CPUs fixed at 64. Dashed lines indicate linear fits to the data.}
    \label{fig:linear}
\end{figure}

Next, we assess parallel scalability using weak- and strong-scaling metrics. In the weak-scaling tests, the number of particles per node (64 CPU cores per node) is fixed while both the node count and the total system size increase proportionally. We define weak-scaling efficiency as \(\eta^{\text{weak}}(\zeta)=t_{\min}/t(\zeta)\), where \(\zeta\) is the number of nodes, \(t(\zeta)\) is the runtime on \(\zeta\) nodes, and \(t_{\min}=t(1)\) is the runtime on a single node. Values closer to 1 indicate better weak scalability. In our experiments, the per-node workload is fixed at \(N_{\text{node}}=345984\) particles, and the node count increases from \(1\) to \(32\) in powers of two. CPU time per step and weak-scaling results are shown in Fig.~\ref{fig:weak}, where the total time for electrostatic calculations is decomposed into long-range and short-range components. For RBSOG and RBE, we use \(P=128\) and \(P=512\), respectively, as these settings yield roughly comparable accuracy in the earlier tests. On a single node, RBSOG delivers an \(18.3\times\) speedup over PPPM for long-range computations. At \(\zeta=32\), the speedup increases to \(38.4\times\) for long-range calculations and \(12.5\times\) for total runtime, with more than \(2\times\) higher weak-scaling efficiency in both components. Figures~\ref{fig:weak}(A) and (D) further show that RBSOG preserves the same weak scaling of RBE on the long-range component, while reducing CPU time by approximately a factor of $4$.

\begin{figure}[ht]
    \centering
    \includegraphics[width=1.0\textwidth]{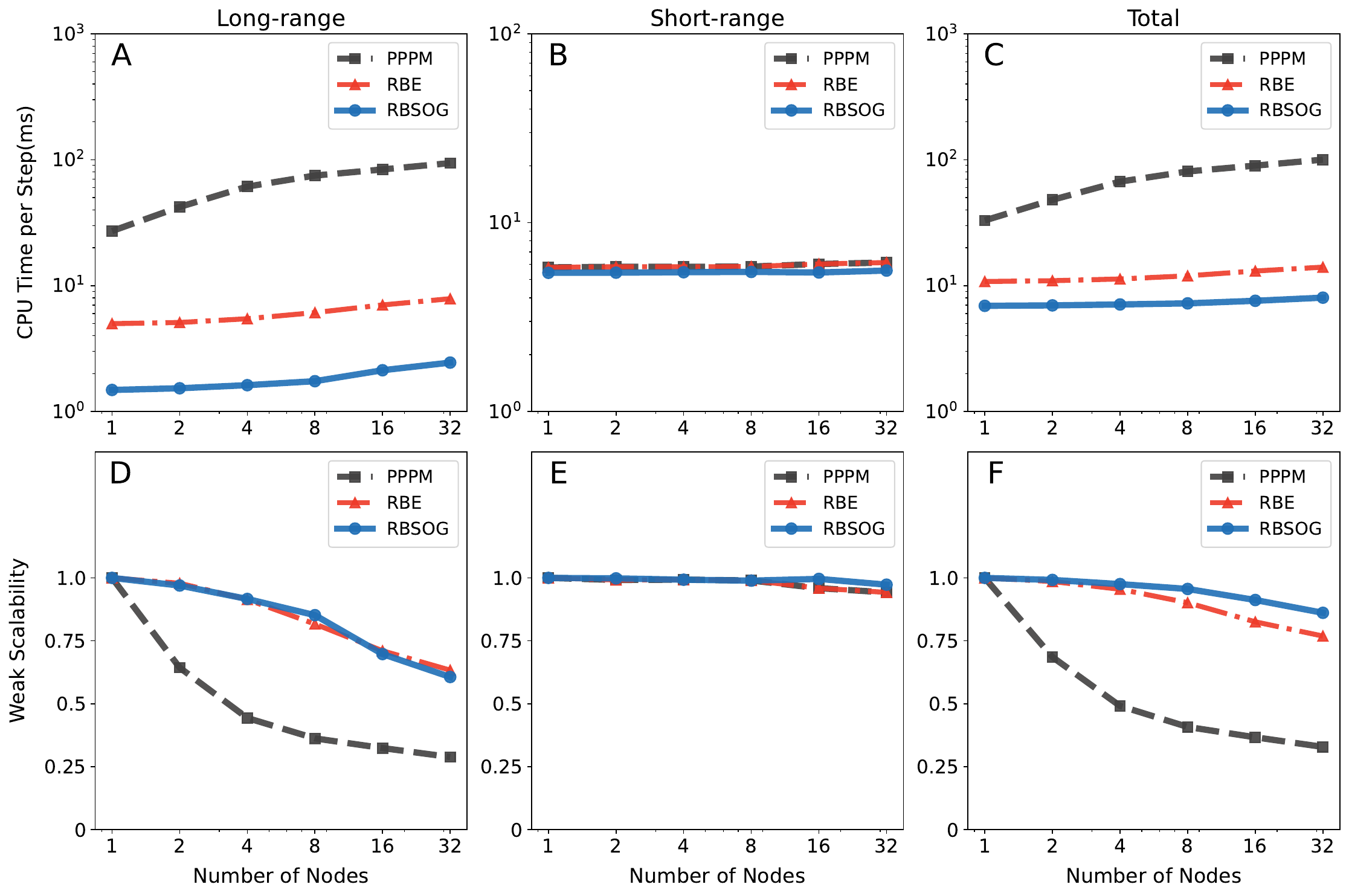}
    \caption{CPU time per step (A--C) and weak-scaling performance (D--F) for long-range, short-range, and total electrostatic interactions. Results are shown for RBSOG ($P=128$), RBE ($P=512$), and PPPM ($\Delta=10^{-5}$).}
    \label{fig:weak}
\end{figure}

Strong scaling is a complementary metric of parallel efficiency. Unlike weak scaling, strong scaling examines whether runtime decreases proportionally with increasing node count at fixed system size. We define the strong-scaling efficiency as $\eta^{\text{strong}}(\zeta)=\zeta_{\min} t_{\min}/(\zeta\, t(\zeta))$ where values closer to 1 indicate better strong scaling. Here, we fix the system size at \(N=11071488\) particles and increase the number of nodes (64 cores per node) from 1 to 32 in powers of two. The results are shown in Fig.~\ref{fig:strong}. For RBSOG and RBE, we again use \(P=128\) and \(P=512\), respectively. As expected, the FFT-based PPPM method relies on dense global communication, and its strong-scaling efficiency degrades sharply as node count increases. In contrast, RBSOG maintains more than \(90\%\) strong-scaling efficiency even at \(\zeta=32\). On a single node, RBSOG achieves a \(27.9\times\) speedup over PPPM for long-range computation and \(9.0\times\) for total runtime. At \(\zeta=32\), the speedups increase to \(53.1\times\) for the long-range part and \(15.0\times\) for total runtime, with approximately \(2\times\) higher strong scaling in both cases. RBE also shows reasonable scalability, but its long-range and total runtimes are about 3.5 and 1.5 times higher than those of RBSOG, respectively, mainly due to its requirement for a larger batch size. These results indicate that RBSOG is well suited for large-scale NPT simulations on modern clusters.

\begin{figure}[ht]
    \centering
\includegraphics[width=1.0\textwidth]{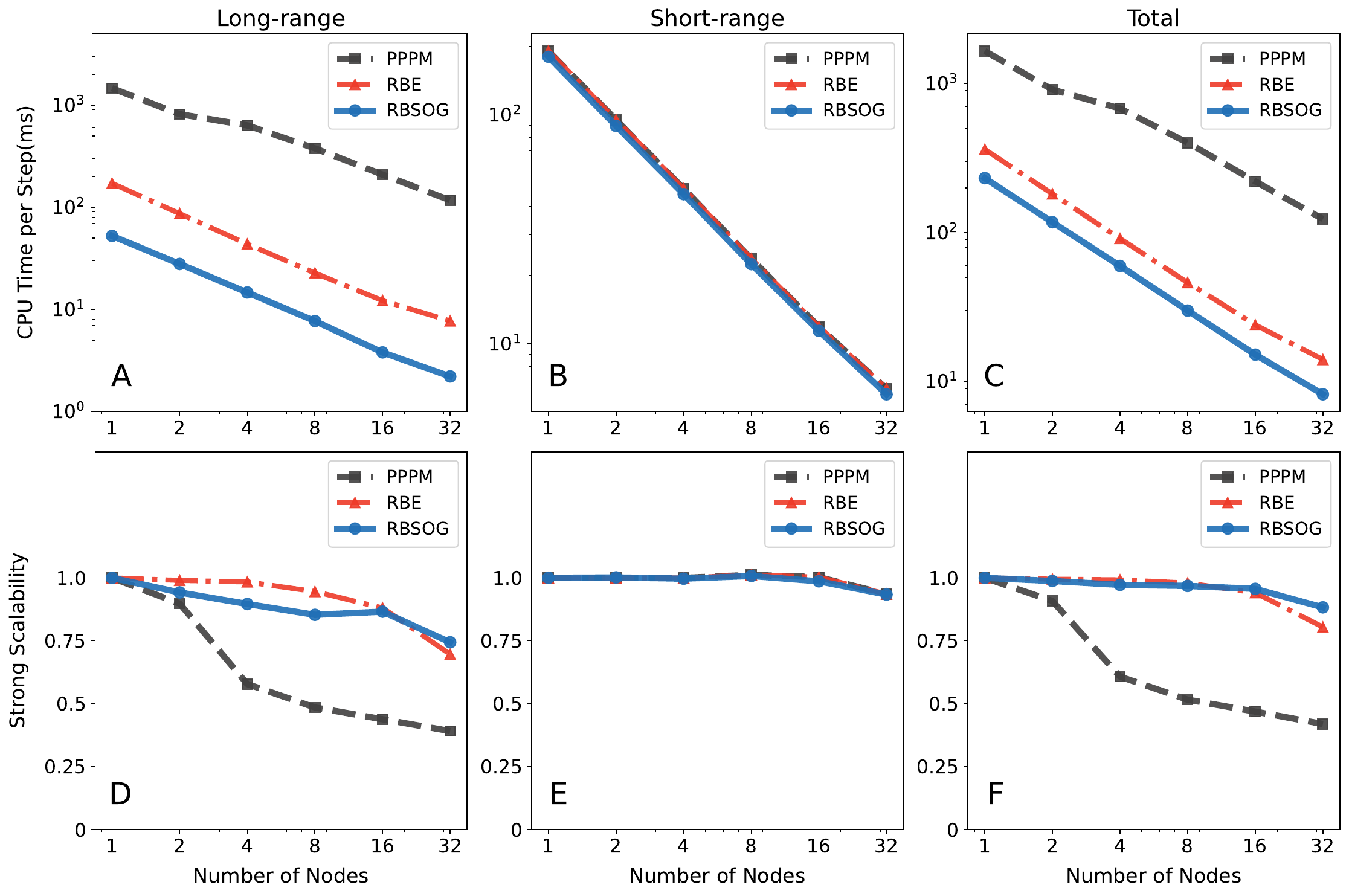}
    \caption{CPU time per step (A--C) and strong-scaling performance (D--F) for long-range, short-range, and total electrostatic interactions. Results are shown for RBSOG ($P=128$), RBE ($P=512$), and PPPM ($\Delta=10^{-5}$).}
    \label{fig:strong}
\end{figure}

\section{Conclusions}
\label{sec::Conclusion}

In this work, we developed an RBSOG method to accelerate MD simulations of charged systems in the NPT ensemble. The method introduces an SOG splitting of the pressure-related $1/r^3$ kernel and combines random batch importance sampling with a measure-recalibration strategy, so that Fourier modes can be reused for both radial and non-radial pressure components. This design yields an unbiased pressure estimator with substantially reduced sampling variance and mitigates pressure artifacts caused by cutoff discontinuities in traditional Ewald-based treatments, while preserving near-optimal $O(N)$ complexity when $P=O(1)$. On the theoretical side, we provide pressure-decomposition error estimates, consistency and variance analysis of the stochastic approximation, and convergence analysis of RBSOG-based MD. Numerical experiments on bulk water, LiTFSI ionic liquids, and DPPC membranes show that RBSOG accurately reproduces key structural and dynamical observables using small batch sizes ($P\sim 100$). In large-scale benchmarks up to $10^7$ atoms on 2048 CPU cores, RBSOG delivers about an order-of-magnitude speedup over PPPM in electrostatic calculations for NPT simulations, together with a consistent $4\times$ variance reduction relative to RBE at the same batch size and excellent weak/strong scalability. These results support RBSOG as a practical and efficient option for large-scale NPT simulations across a broad range of applications. Our future work will focus on GPU implementation and extending the method to quasi-2D systems~\cite{gan2025random,gao2025accurate}.

\section*{Acknowledgments}
This work is funded by the National Natural Science Foundation of China (grants No. 12426304, 12401570, 12325113 and 125B2023), and the Science and Technology Commission of Shanghai Municipality (grant No. 23JC1402300). The work of J. L. is partially supported by the China Postdoctoral Science Foundation (grant No. 2024M751948). The authors also acknowledge the support from the SJTU Kunpeng \& Ascend Center of Excellence.

\section*{Declarations}
\subsection*{Conflict of interest}
The authors declare that they have no conflict of interest.

\subsection*{Data availability}
The data that support the findings of this study are available from the corresponding author upon reasonable request.

\appendix
\section{Formulation of the energy and force}
\label{app::force_energy}
With the $u$-series splitting in Eq.~\eqref{eq::potential_split}, the Coulomb energy $U_{\text{coul}}$ is decomposed into short-range, long-range, and self-energy contributions, 
\begin{equation}\label{eq::couldecomp}
U_{\text{coul}}=U_{\text{coul}}^\mathcal{N}+U_{\text{coul}}^\mathcal{F}-U_{\text{self}}.
\end{equation}
The short-range part is given by
\begin{equation}
\label{eq::U_short}
	U_{\text{coul}}^\mathcal{N}(\{\bm{r}_i\},\bm{h})=\dfrac{1}{2}\sum_{\bm{n}\in\mathbb{Z}^3}\!^\prime\sum_{i,j=1}^{N}q_iq_j\mathcal{N}_b^{\sigma}(|\bm{r}_{ij}+\bm{h}\bm{n}|),
\end{equation}
and the long-range part admits the Fourier-space representation
\begin{equation}
\label{eq::U_long}
	U_{\text{coul}}^\mathcal{F}(\{\bm{r}_i\},\bm{h})
=\frac{1}{2\det(\bm{h})}\sum_{\bm{k}\neq \bm{0}}
\widehat{\mathcal{F}}_b^\sigma(\bm{k})\,\left|\rho(\bm{k})\right|^2,
\end{equation}
where 
\begin{equation}
\widehat{\mathcal{F}}_b^\sigma(\boldsymbol{k})=\sum_{\ell=0}^{M-1}\pi^{3/2}w_\ell s_\ell^{3}e^{-s_\ell^2 k^2/4}
\end{equation} 
denotes the Fourier transform of $\mathcal{F}_b^\sigma(\bm{r})$. 
The self-energy term removes the nonphysical self-interaction,
\begin{equation}
	U_{\text{self}}=\frac{\ln b}{\sqrt{2\pi}\sigma}\left(\omega+\frac{1-b^{-(M-1)}}{b-1}\right)\sum_{i=1}^{N}q_i^2.
\end{equation}
The Coulomb force $\bm{F}_{i,\text{coul}}=\bm{F}_{i,\mathrm{coul}}^{\mathcal{N}}+\bm{F}_{i,\mathrm{coul}}^{\mathcal{F}}$ follows from
$\bm{F}_i=-\nabla_{\bm{r}_i}U$, i.e., from differentiating
Eqs.~\eqref{eq::U_short}--\eqref{eq::U_long}. This yields
\begin{equation}
\label{eq::force_SOG}
\begin{aligned}
\bm{F}_{i,\mathrm{coul}}^{\mathcal{N}}
&=\sum_{\bm{n}\in\mathbb{Z}^3}\!^\prime\sum_{j=1}^{N}
q_i q_j\,
\left(\frac{1}{|\bm{r}_{ij}+\bm{h}\bm{n}|^3}-\sum_{\ell=0}^{M-1}\frac{2w_{\ell}}{s_{\ell}^2}e^{-|\bm{r}_{ij}+\bm{h}\bm{n}|^2/s_{\ell}^2}\right)\,
\bigl(\bm{r}_{ij}+\bm{h}\bm{n}\bigr)
\end{aligned}
\end{equation}
and
\begin{equation}
\bm{F}_{i,\mathrm{coul}}^{\mathcal{F}}=-\frac{q_i}{\det(\bm{h})}\sum_{\bm{k}\neq \bm{0}}
\bm{k}\widehat{\mathcal{F}}_b^\sigma(\bm{k})\,
\;
\text{Im}\!\left(e^{-i\bm{k}\cdot\bm{r}_i}\rho(\bm{k})\right).
\end{equation}
Finally, $U_{\mathrm{self}}$ is a constant when charges are fixed (i.e., in the absence of charge transfer or charge regularization effects), and therefore does not contribute to the force.

\section{Random batch approximation of force}\label{app::RBForce}
In NPT simulations, forces are also required to evolve the equations of motion (e.g., the Langevin dynamics \eqref{eq::Langevin_NPT}) at each MD step. In our implementation, the short-range force $\bm{F}_{i,\mathrm{coul}}^{\mathcal{N}}$ is truncated at a cutoff radius $r_c$ and evaluated in $O(N)$ cost using a neighbor list. For the long-range force contribution $\bm{F}_{i,\mathrm{coul}}^{\mathcal{F}}$, we adopt a random batch sampling strategy.

We start from the long-range force in \eqref{eq::force_SOG}. Since the SOG representation $\widehat{\mathcal{F}}_{b}^{\sigma}(\bm{k})$ is absolutely summable over the reciprocal lattice, the force can be written as an expectation over Fourier modes. We then approximate $\bm{F}_{i,\mathrm{coul}}^{\mathcal{F}}$ via importance sampling, following the original RBSOG idea developed for the NVT ensemble~\cite{liang2023random}. Concretely, at each step we draw a batch of nonzero integer vectors $\{\bm{m}_p\}_{p=1}^{P}\subset\mathbb{Z}^3\setminus\{\bm{0}\}$ from the discrete distribution proportional to $|\bm{k}|^2\widehat{\mathcal{F}}_{b}^{\sigma}(\bm{k})$ with $\bm{k}=2\pi\bm{h}^{-\top}\bm{m}$, i.e.,
\begin{equation}
\label{eq::prob_p}
	\mathscr{P}(\bm{m}) =\frac{1}{S}\,|2\pi\bm{h}^{-\top}\bm{m}|^2 \,\widehat{\mathcal{F}}_{b}^{\sigma}(2\pi\bm{h}^{-\top}\bm{m}),
\end{equation}
where $S$ is the normalization constant that makes $\mathscr{P}(\bm{m})$ a probability measure. With the sampled set $\{\bm{m}_p\}_{p=1}^{P}$ and $\bm{k}_p=2\pi\bm{h}^{-\top}\bm{m}_p$, the long-range force is approximated by
\begin{equation}
\label{eq::f_approx}
\bm{F}_{i,\mathrm{coul}}^{\mathcal{F}*}=-\frac{S}{P}\frac{q_i}{ \det(\bm{h})} \sum_{p=1}^P \frac{\boldsymbol{k}_p \operatorname{Im}\!\left(e^{-\mathrm{i} \boldsymbol{k}_p \cdot \boldsymbol{r}_i} \rho\!\left(\boldsymbol{k}_p\right)\right)}{\left|\boldsymbol{k}_p\right|^2}.
\end{equation}

Compared with sampling directly from the SOG kernel, the proposal \eqref{eq::prob_p} leads to better-behaved importance weights. In particular, as $|\boldsymbol{k}|\to 0$, one has
\begin{equation}
\left|\boldsymbol{k}\,\operatorname{Im}\!\left(e^{-\mathrm{i} \boldsymbol{k} \cdot \boldsymbol{r}_i} \rho(\boldsymbol{k})\right)\right| \sim |\boldsymbol{k}|^2,
\end{equation}
so the summand in \eqref{eq::f_approx} is approximately constant over the long-wavelength modes. Since these modes are especially important in periodic systems, the variance of the stochastic estimator is reduced. Similar to the pressure tensor evaluation, \eqref{eq::f_approx} can be computed for all particles in $O(PN)$ operations per step, which becomes $O(N)$ when $P=O(1)$.

\section{Verification of the virial theorem}\label{app::virialtheorem}
In Theorem~\ref{thm::virial}, we assume a consistent parameter setting for the decompositions of the Coulomb potential and pressure, namely
$b=\widetilde{b}$, $\sigma=\widetilde{\sigma}$, and $M=\widetilde{M}$.
Under this setting, direct calculation from Eqs.~\eqref{eq::P_long} and \eqref{eq::U_long} implies
$w_{\ell}=\widetilde{w}_{\ell}\widetilde{s}_{\ell}^2/2$ and $s_{\ell}=\widetilde{s}_{\ell}$, and hence
\begin{equation}
\label{eq::UF}
	\text{tr}\left(\bm{P}^{\mathcal{F}}_{\text{coul}}\right)-\frac{U_{\text{coul}}^\mathcal{F}}{\det(\bm{h})}
	=\frac{\pi^{3/2}}{4\det(\bm{h})^2}\sum_{\bm{k}\neq 0}\sum_{\ell=0}^{M-1}\widetilde{w}_\ell \widetilde{s}_\ell^5
	e^{-\widetilde{s}_\ell^2|\bm{k}|^2/4}|\rho(\bm{k})|^2\left[2-\frac{\widetilde{s}_\ell^2}{2}|\bm{k}|^2\right].
\end{equation}
Similarly, using Eqs.~\eqref{eq::P_short} and \eqref{eq::U_short}, we obtain
\begin{equation}
\label{eq::UN}
\begin{aligned}
-\text{tr}\left(\bm{P}^{\mathcal{N}}_{\text{coul}}\right)+\frac{U_{\text{coul}}^\mathcal{N}}{\det(\bm{h})}
=\frac{1}{2\det(\bm{h})}\sum_{\bm{n}}\!^\prime\sum_{i,j}q_iq_j
\left[\sum_{\ell=0}^{M-1}\widetilde{w}_\ell
e^{-|\bm{r}_{ij}+\bm{h}\bm{n}|^2/\widetilde{s}_\ell^2}
\left(\frac{\widetilde{s}_{\ell}^2}{2}+|\bm{r}_{ij}+\bm{h}\bm{n}|^2\right)\right].
 \end{aligned}
\end{equation}

To connect Eq.~\eqref{eq::UF} and Eq.~\eqref{eq::UN}, we rewrite the real-space sum in Eq.~\eqref{eq::UN} using discrete Fourier sums.
Let $f(\bm r)$ be a Schwartz function and denote its Fourier transform by $\hat f(\bm k)$.
By the Poisson summation formula (with $\bm k=2\pi\bm h^{-\top}\bm m$, $\bm m\in\mathbb Z^3$), we have
\begin{equation}\label{eq::Poisson}
\sum_{\bm{n}}\!^\prime\sum_{i,j}q_iq_j f(\bm{r}_{ij}+\bm{h}\bm{n})
=\frac{1}{\det(\bm h)}\sum_{\bm{k}}\hat f(\bm{k})|\rho(\bm{k})|^2-\sum_{i=1}^{N}q_{i}^2 f(\bm{0}).
\end{equation}
If the system is charge neutral, then $\rho(\bm 0)=0$ and the $\bm k=\bm 0$ term vanishes, so one may equivalently sum over $\bm k\neq \bm 0$. We define
\[
f_1(\bm r)=\sum_{\ell=0}^{M-1}\widetilde{w}_\ell e^{-|\bm r|^2/\widetilde{s}_\ell^2}\,|\bm r|^2,
\qquad
f_2(\bm r)=-\sum_{\ell=0}^{M-1}\frac{\widetilde{w}_\ell \widetilde{s}_{\ell}^2}{2} e^{-|\bm r|^2/\widetilde{s}_\ell^2}.
\]
Their Fourier transforms are
\begin{equation}
\begin{aligned}
\hat f_1(\bm{k})=\frac{\pi^{3/2}}{2}\sum_{\ell=0}^{M-1}\widetilde{w}_\ell \widetilde{s}_\ell^5
e^{-\widetilde{s}_\ell^2|\bm{k}|^2/4}\left[3-\frac{\widetilde{s}_\ell^2}{2}|\bm{k}|^2\right],\qquad\hat f_2(\bm{k})=-\frac{\pi^{3/2}}{2}\sum_{\ell=0}^{M-1}\widetilde{w}_\ell \widetilde{s}_\ell^5
e^{-\widetilde{s}_\ell^2|\bm{k}|^2/4}.
\end{aligned}
\end{equation}
Applying Eq.~\eqref{eq::Poisson} to $f_1$ and $f_2$ and substituting the results into Eq.~\eqref{eq::UN} yields
\begin{equation}
\label{eq::final_virial}
	\text{tr}\left(\bm{P}^{\mathcal{N}}_{\text{coul}}\right)-\frac{U_{\text{coul}}^\mathcal{N}}{\det(\bm{h})}
	=-\text{tr}\left(\bm{P}^{\mathcal{F}}_{\text{coul}}\right)+\frac{U_{\text{coul}}^\mathcal{F}}{\det(\bm{h})}-\frac{U_{\text{self}}}{\det(\bm{h})}.
\end{equation}
Rearranging terms in \eqref{eq::final_virial} gives Eq.~\eqref{eq::Virli_flexible} in Theorem~\ref{thm::virial}.

Moreover, to ensure force continuity, we enforce Eq.~\eqref{eq::2.22}, i.e.
\begin{equation}
    \label{eq::correlation_kernels}
    \begin{aligned}
        -\frac{d}{dr}\left[\frac{1}{r}-\mathcal{F}_{b}^{\sigma}(r)\right]\bigg|_{r=r_c}=-\frac{1}{r_c^2}+2r_c\sum_{\ell=0}^{M-1}w_{\ell}s_{\ell}^{-2}e^{-r_c^2/s_{\ell}^2}=-r_c\left[\frac{1}{r_c^3}-\widetilde{\mathcal{F}}_{b}^{\sigma}(r_c)\right]
        =0,
    \end{aligned}
\end{equation}
where the second equality follows from the weight relation $w_{\ell}=\widetilde{w}_{\ell}\widetilde{s}_{\ell}^2/2$ and $s_{\ell}=\widetilde{s}_{\ell}$.
Since Eq.~\eqref{eq::correlation_kernels} is equivalent to the pressure continuity condition Eq.~\eqref{eq::continuity}, the consistent truncation implies that the narrowest-Gaussian matching gives $\omega=\widetilde{\omega}$.
Therefore, both force and pressure are continuous under the same set of splitting parameters. This finishes the proof of Theorem~\ref{thm::virial}.

\section{Metropolis--Hastings algorithm}\label{app::MHsamp}
This section describes how we use the MH algorithm to draw samples from the proposal distribution \(\mathscr{P}^{\mathrm{r}}(\bm{k})\) defined in Eq.~\eqref{eq::samplingP}. Since Fourier modes lie on the reciprocal lattice, we write
\(\bm{k}=2\pi \bm{h}^{-\top}\bm{m}\),
where \(\bm{m}\in \mathbb{Z}^3\setminus\{\bm{0}\}\) is an integer vector. We therefore define a discrete target distribution on \(\bm{m}\) by
\(\mathscr{P}(\bm{m}) := \mathscr{P}^{\mathrm{r}}(2\pi \bm{h}^{-\top}\bm{m})\),
and generate samples \(\{\bm{m}_p^{\mathrm{r}}\}_{p=1}^{P^{\mathrm{r}}}\) from this distribution.

At each MH step, we propose a new integer vector \(\bm{m}^*\) as follows. We first draw three samples separately from the (continuous) normal distribution
\(u_\xi \sim \mathcal{N}\!\left(0,\; 1/(2\pi^2\widetilde{s}_0^2)\right)\) for \(\xi\in\{x,y,z\}\),
and set \(\bm{u}=(u_x,u_y,u_z)^\top\). We then map \(\bm{u}\) to the integer lattice via the proposal 
\(\bm{m}^*=\mathrm{round}(\bm{h}^{\top}\bm{u})\),
where \(\mathrm{round}(\cdot)\) is applied componentwise. To determine acceptance of $\bm{m}^*$, we evaluate the proposal distribution (which depends only on $\bm{m}^*$):
\begin{equation}
    \label{eq::proposal}
q(\bm{m}^*|\bm{m}) = q(\bm{m}^*)
 =\frac{(\sqrt{\pi} \widetilde{s}_0)^3}{\det(\bm{h})}
 \int_{\Omega_{\bm{m}^{*}}}
e^{-\pi^2 \widetilde{s}_0^2 \bm{\gamma}^{\top}(\bm{h}^{\top}\bm{h})^{-1}\bm{\gamma}}
 \mathrm{d}\bm{\gamma},
\end{equation}
where $\bm{m}$ is the current state of the Markov chain and $\Omega_{\bm{m}^{*}}=\prod_{\xi\in\{x,y,z\}}[m^*_\xi-\frac{1}{2},m^*_\xi+\frac{1}{2}]$. The MH acceptance probability is then given by
\begin{equation}
    \label{eq::acc_Gaussian}
    \mathrm{Acc}(\bm{m}^*, \bm{m})
    :=\min \left\{1,\;
    \frac{\mathscr{P}\left(\bm{m}^*\right) q\left(\bm{m}| \bm{m}^*\right)}
         {\mathscr{P}(\bm{m})\, q\left(\bm{m}^*|\bm{m}\right)}
    \right\}.
\end{equation}
In our setting, the acceptance rate is typically high because the proposal $q$ is designed to be close to the target $\mathscr{P}$. If the proposal is accepted, we set \(\bm{m}\leftarrow \bm{m}^*\); otherwise, we keep the current state. Repeating this MH step for \(P^{\mathrm{r}}\) iterations yields samples \(\{\bm{m}_p^{\mathrm{r}}\}_{p=1}^{P^{\mathrm{r}}}\), and we finally obtain \(\bm{k}_p^{\mathrm{r}}=2\pi \bm{h}^{-\top}\bm{m}_p^{\mathrm{r}}\), which follow the target distribution \(\mathscr{P}^{\mathrm{r}}(\bm{k})\).

\section{The Debye--H\"{u}ckel theory}
\label{app::DH_theory}

This section derives the estimate of $|\rho(\bm{k})|^4$ in Eq.~\eqref{eq::DH_upperbound} within the DH regime~\cite{levin2002electrostatic}. The DH theory considers the simplest model of an electrolyte solution confined to the simulation cell, where all $N$ ions are idealized as hard spheres of diameter $a$ carrying charge $\pm q$ at their centers. 
Let us fix one ion of charge $+q$ at the origin and consider the distribution of the other ions around it.
Inside the excluded region $0<r<a$, there are no other ions, hence the electrostatic potential satisfies $\nabla^2\phi=0$.
For $r\ge a$, the number density of each species is described by the Boltzmann distribution
$n_{\pm}(\bm r)=\rho_\infty \exp(\mp \beta q \phi(\bm r))$,
where $\rho_\infty=N/(2V)$ and $\beta=(k_B T)^{-1}$.
Under weak coupling, linearizing the Poisson--Boltzmann equation gives
\begin{equation}
\label{eq::Debye_revised}
-\varepsilon \nabla^2\phi(\bm{r})
= q\big(n_+(\bm{r}) - n_-(\bm{r})\big)
\approx - 2\beta q^2 \rho_\infty \,\phi(\bm{r})
= -\varepsilon \kappa^2 \phi(\bm{r}),
\end{equation}
with $\kappa^2=2\beta q^2\rho_\infty/\varepsilon$.
Solving \eqref{eq::Debye_revised} yields
\begin{equation}
\phi(\bm{r})=
\begin{cases}
\dfrac{q}{4\pi\varepsilon r}-\dfrac{q\kappa}{4\pi\varepsilon(1+\kappa a)}, & 0<r<a,\\[0.8em]
\dfrac{q\,e^{-\kappa(r-a)}}{4\pi\varepsilon r(1+\kappa a)}, & r\ge a,
\end{cases}
\end{equation}
and therefore the net charge density for $r>a$ is
$\rho_{>}(r)=-\varepsilon \kappa^2 \phi(r)$. Using $|\rho(\bm k)|^2=\sum_{i,j}q_i q_j e^{i\bm k\cdot(\bm r_i-\bm r_j)}$ and the DH solution of $\rho_{>}$, we have
\begin{equation}
\label{eq::rho_DH}
\begin{aligned}
|\rho(\bm{k})|^4
&
=\left|\sum_{i,j=1}^{N}q_iq_je^{i\bm{k}\cdot(\bm{r}_i-\bm{r}_j)}\right|^2\\
&\approx
\left(\sum_{i=1}^{N} q_i\left[q_i+\int_{\mathbb{R}^3\setminus \mathring{B}(a)}
\rho_{>}(\bm r)\,e^{i\bm k\cdot \bm r}\,d\bm r\right]\right)^2\\
&=
\left(\sum_i q_i^2\left[1-
\frac{\frac{\kappa}{k}\sin(ka)+\cos(ka)}{(1+\kappa a)\,(1+k^2/\kappa^2)}
\right]\right)^2.
\end{aligned}
\end{equation}
The bracketed factor is uniformly bounded for all $k\ge 0$ (it tends to $0$ as $k\to 0$ and to $1$ as $k\to\infty$),
hence there exists a constant $C$ such that
\begin{equation}
|\rho(\bm k)|^4 \le C \left(\sum_{i=1}^{N} q_i^2\right)^2 \le C N^2 q_{\max}^4,
\end{equation}
where $q_{\text{max}}:=\max_{i}|q_{i}|$. This completes the DH-based estimate.

	
\end{document}